\documentclass[
reprint,
superscriptaddress,
 amsmath,amssymb,
 aps,
floatfix,
longbibliography
]{revtex4-1}

\usepackage{graphicx}
\usepackage{dcolumn}
\usepackage{bm}
\usepackage{dsfont}
\usepackage{amsthm}
\usepackage{chngcntr}
\usepackage{apptools}
\usepackage{amsmath}
\usepackage{mathtools}
\usepackage{graphicx}
\graphicspath{ {./images/} }
\usepackage{siunitx}
\usepackage[margin=1in,footskip=0.25in]{geometry}
\usepackage{blindtext}
\usepackage{enumitem}
\usepackage{xcolor}
\usepackage[makeroom]{cancel}

\usepackage{placeins}
\AtAppendix{\counterwithin{lemma}{section}}
\usepackage{longtable} 

\usepackage{hyperref}
\usepackage[capitalise]{cleveref}
\crefname{figure}{Fig.}{Fig.}

\newtheorem{lemma}{Lemma}
\theoremstyle{definition}
\newtheorem{definition}{Definition}
\theoremstyle{theorem}

\theoremstyle{definition}
\newtheorem{example}{Example}[section]

\newenvironment{sloppypar*}{\sloppy\ignorespaces}{\par}
\allowdisplaybreaks 

\newcommand{\bma} {\begin{pmatrix}}
\newcommand{\ema} {\end{pmatrix}}

\newcommand{\qubitn}{\mathsf{Q}}

\DeclareMathAlphabet{\mathdutchcal}{U}{dutchcal}{m}{n}

\newcommand{\be}{\begin{equation}}
\newcommand{\ee}{\end{equation}}
\newcommand{\bp}{\begin{pmatrix}}
\newcommand{\ep}{\end{pmatrix}}
\newcommand{\ben}{\begin{enumerate}}
\newcommand{\een}{\end{enumerate}}

\newcommand{\eps}{\epsilon}

\newcommand{\Lcutoff}{\Lambda^\text{min}} 
\newcommand{\Ucutoff}{\Lambda^\text{max}} 

\begin{document}

\title{Quantum Simulation of Second-Quantized Hamiltonians in Compact Encoding}

\author{William M. Kirby}
\email{william.kirby@tufts.edu}
\affiliation{Department of Physics and Astronomy, Tufts University, Medford, MA 02155, USA
}

\author{Sultana Hadi}
\affiliation{Department of Physics and Astronomy, Tufts University, Medford, MA 02155, USA
}

\author{Michael Kreshchuk}
\affiliation{Department of Physics and Astronomy, Tufts University, Medford, MA 02155, USA
}
\affiliation{Physics Division, Lawrence Berkeley National Laboratory, Berkeley, CA 94720, USA}

\author{Peter J. Love}
\affiliation{Department of Physics and Astronomy, Tufts University, Medford, MA 02155, USA
}

\begin{abstract}
We describe methods for simulating general second-quantized Hamiltonians using the compact encoding, in which qubit states encode only the occupied modes in physical occupation number basis states.
These methods apply to second-quantized Hamiltonians composed of a constant number of interactions, i.e., linear combinations of ladder operator monomials of fixed form.
Compact encoding leads to qubit requirements that are optimal up to logarithmic factors.
We show how to use sparse Hamiltonian simulation methods for second-quantized Hamiltonians in compact encoding, give explicit implementations for the required oracles, and analyze the methods.
We also describe several example applications including the free boson and fermion theories, the $\phi^4$-theory, and the massive Yukawa model, all in both equal-time and light-front quantization.
Our methods provide a general-purpose tool for simulating second-quantized Hamiltonians, with optimal or near-optimal scaling with error and model parameters.
\end{abstract}

\pacs{Valid PACS appear here}
\maketitle

\tableofcontents

\section{Introduction}
\label{intro}

We describe a framework for simulating second-quantized Hamiltonians on quantum computers.
Hamiltonians in second-quantization are ubiquitous in quantum chemistry, many-body physics, and quantum field theory, all of which are target applications for quantum simulation.
Fermionic Hamiltonians with fixed particle number admit simple encodings in the Pauli basis~\cite{jordan28a,bravyi02a,seeley12a}, and these have been the focus of many quantum simulation experiments to date~\cite{DuH2NMR,lanyon2010towards,peruzzo14a,wang2015quantum,omalley16a,santagati18a,shen2017quantum,paesani2017,kandala17a,hempel18a,dumitrescu18a,colless18a,nam20a,kokail19a,kandala19a,google20a}.
However, second-quantized Hamiltonians are \emph{sparse}~--- they have only polynomially-many nonzero entries per row or column~--- as long as they have polynomially-many terms.
This makes them appropriate for simulation using methods developed for sparse Hamiltonians~\cite{aharonov03a,childs03a,berry07a,childs10a,berry12a,berry14a,berry15a,berry15b,low17a,low19a,berry20a}.

The second-quantized Hamiltonians we consider are given as polynomials in ladder operators acting on occupation number states (\emph{Fock states}).
The main idea is to extend the \emph{compact encoding} previously studied in~\cite{aspuru2005simulated,toloui2013cimatrix,kreshchuk20a}, which only stores information about occupied modes in a given Fock state.
The application of sparse simulation techniques to electronic structure Hamiltonians was previously studied in~\cite{babbush16a,babbush17a}, and these papers use a special case of the compact encoding (which they call ``compressed representation").
In \cref{prior_work}, we will compare the overall cost of simulation using our algorithm to those of~\cite{babbush16a,babbush17a}.

The compact encoding is to be contrasted with \emph{direct encodings}, which store information about all physical modes, whether they are occupied or not.
The Jordan-Wigner and Bravyi-Kitaev encodings commonly used in quantum algorithms for quantum chemistry are examples of direct encodings~\cite{jordan28a,bravyi02a,seeley12a}.
Compact encodings are suitable for Hamiltonians that are sparse in the occupation number basis.
In a sparse Hamiltonian, the number of nonzero elements in each row or column scales polynomially with the problem size, and therefore polylogarithmically with Hamiltonian dimension.
The compact encoding permits efficient sparsity-based state preparation and time-evolution methods~\cite{aspuru2005simulated,toloui2013cimatrix,kreshchuk20a}.

The methods we develop in this paper are motivated by simulation of quantum field theory.
In particular, we will focus on the case of Hamiltonians expressed in the plane wave momentum basis as our main example, since it illustrates the key techniques of our method.
We use the fact that such Hamiltonians can be expressed as sums of \emph{interactions}, where an interaction is a sum of ladder operator monomials that only differ in their momentum quantum numbers.
The sum within each interaction runs over all assignments of momenta that conserve the total momentum (see \cref{hamiltonian} for details).

In Sections  \ref{main_results_sec}-\ref{applications} we choose to define multi-particle states using the plane wave momentum basis, as is typically done in quantum field theory, because this example is sufficiently complex to capture the main considerations.
However, our method extends straightforwardly to Hamiltonians where the sums within interactions run over quantum numbers other than plane wave momenta, as long as the number of distinct interactions in the Hamiltonian is polynomial in the system parameters (such as momentum cutoffs).
These cases include a wide range of theories in quantum chemistry, condensed matter physics, and quantum field theory, including basis light-front quantization~\cite{varybasis,kreshchuk2020light,kreshchuk2020blfq}.
How to extend our methods beyond the plane wave momentum basis is explained in \cref{beyond_momentum}.

\section{Main results for plane wave momentum basis}
\label{main_results_sec}

Algorithms for simulating general sparse Hamiltonians access the Hamiltonian via oracle unitaries that are queried (applied) to provide the locations and values of the nonzero Hamiltonian matrix elements~\cite{aharonov03a,childs03a,berry07a,childs10a,berry12a,berry14a,berry15a,berry15b,low17a,low19a,berry20a}.
If we want to apply such algorithms to a second-quantized Hamiltonian in compact encoding, then we have to provide two main additional components.
First, we need to explicitly construct the oracle unitaries for the specific Hamiltonian of interest, as sequences of primitive gates.
We will show how to do this in \cref{enumerator,matrix_element}, decomposing the oracle unitaries into qubit operations that are log-local in the problem parameters, which for us will be momentum cutoffs, since we focus on the example of the plane wave momentum basis.
The log-local operations can then themselves be decomposed into primitive gates from any desired gate set with only polynomial overhead.

Second, the general sparse Hamiltonian methods assume that the oracles act directly upon row and column indices (encoded in qubit states) of the Hamiltonian~\cite{aharonov03a,childs03a,berry07a,childs10a,berry12a,berry14a,berry15a,berry15b,low17a,low19a,berry20a}.
We instead want methods that act directly upon compact-encoded Fock states, because the physical meaning of such states can be directly read out, which ultimately permits efficient implementation of the oracle unitaries as well as of observables.
However, unlike simply labeling the rows and columns of the Hamiltonian by sequential binary numbers, the set of bitstrings corresponding to compact-encoded Fock states is not simple to characterize or enumerate.
These bitstrings label computational basis states that span the subspace of qubit Hilbert space that the Hamiltonian acts on.
Therefore, we need to show that when we implement oracle unitaries that act directly on compact-encoded Fock states, the high-level simulation algorithms~\cite{aharonov03a,childs03a,berry07a,childs10a,berry12a,berry14a,berry15a,berry15b,low17a,low19a,berry20a} that use the oracles as their building blocks will still work.
This is explained in \cref{sparsity}.

The overall asymptotic costs of our methods in both qubit and gate counts are summarized in \cref{tab:costs}, for the example of the plane wave momentum basis.
The details of the costs are as follows.
The number of qubits required to encode a Fock state in compact encoding is derived in \cref{encoding}, resulting in the expression in \eqref{qubitn}, which is asymptotically
\begin{equation}
\label{qubitn_main}
    \qubitn=O\bigg(I\log W+I\sum_{j=1}^d\log\left(\Ucutoff_j-\Lcutoff_j\right)\bigg),
\end{equation}
where $I$ is the maximum possible number of occupied modes in a Fock state, $W$ is the maximum possible occupation of any mode, $d$ is the number of spatial dimensions, and $\Lcutoff_j$ and $\Ucutoff_j$ are lower and upper momentum cutoffs in each dimension $j$.
Hence fixing $|\Ucutoff_j|,|\Lcutoff_j|\le\Lambda$ for some overall cutoff $\Lambda$ results in the scaling given in \cref{tab:costs}.
The expression in \eqref{qubitn_main} assumes that the number of qubits required to encode the non-momentum quantum numbers is constant.

In our implementations the cost in log-local gates of the enumerator oracle (the oracle that gives the locations of nonzero matrix elements, defined in \eqref{enumerator_oracle_def}) asymptotically dominates the cost of the matrix element oracle (defined in \eqref{matrix_element_oracle_def}).
These costs are derived in \cref{enumerator} and \cref{matrix_element}, respectively, and result in the expressions \eqref{cost_scaling} and \eqref{matrix_element_cost_scaling}.
The dominant cost is the former, which is
\begin{equation}
\label{oracle_cost}
    O\Big(I^{h}+\Lambda^{dg}\Big)
\end{equation}
exactly as in \cref{tab:costs}, where $h$ is the maximum number of annihilation operators in any interaction in the Hamiltonian and $g$ is the maximum number of creation operators in any interaction in the Hamiltonian.

Finally, if our Hamiltonian is time-independent, then by using qubitization~\cite{low19a} the total number of oracle queries required to simulate time-evolution is
\begin{equation}
    O\left(\tau+\frac{\log(1/\eps)}{\log\log(1/\eps)}\right),
\end{equation}
where $\tau=k\|H\|_\text{max}t$, $k$ is the sparsity of the Hamiltonian $H$, $t$ is the total evolution time, and $\eps$ is the error.
Multiplying by the oracle cost \eqref{oracle_cost} gives the overall asymptotic scaling of the number of log-local gates:
\begin{equation}
\label{total_cost}
    O\Bigg[\bigg(\tau+\frac{\log(1/\eps)}{\log\log(1/\eps)}\bigg)\bigg(I^{h}+\Lambda^{dg}\bigg)\Bigg].
\end{equation}

If instead our Hamiltonian is time-dependent, then by using the method of~\cite{berry20a} the total number of oracle queries required to simulate time-evolution is
\begin{equation}
    O\left(\tau\frac{\log(\tau/\eps)}{\log\log(\tau/\eps)}\right),
\end{equation}
where now $\tau\equiv k\int_{0}^t\|H\|_\text{max}dt$ (without loss of generality taking the starting time to be $t=0$).
Hence the overall log-local gate count for our algorithm is
\begin{equation}
\label{total_cost_time_dep}
    O\Bigg[\bigg(\tau\frac{\log(\tau/\eps)}{\log\log(\tau/\eps)}\bigg)\bigg(I^{h}+\Lambda^{dg}\bigg)\Bigg].
\end{equation}
Suppressing the logarithmic components in either \eqref{total_cost} or \eqref{total_cost_time_dep} gives the expression in \cref{tab:costs}.

\begin{table}[]
    \centering
    \textbf{Parameters:}\\\vspace{0.05in}
    \begin{tabular}{c|c}
        number of spatial dimensions & $d$ \\\hline
        max number of occupied modes & $I$ \\\hline
        max occupancy of a single mode & $W$ \\\hline
        momentum cutoff & $\Lambda$ \\\hline
        max incoming lines in any interaction & $h$ \\\hline
        max outgoing lines in any interaction & $g$ \\\hline
        sparsity & $k$ \\\hline
        max-norm of Hamiltonian & $\|H\|_\text{max}$ \\\hline
        simulation time  & $t$
    \end{tabular}
    \caption{Glossary of parameter definitions.}
    \label{tab:glossary}
\end{table}

\begin{table}[]
    \centering
    \textbf{Costs:}\\\vspace{0.05in}
    \begin{tabular}{c|c}
        qubits to encode Fock state & $O\Big(I\log(W\Lambda^d)\Big)$ \\\hline
        log-local operations for oracle & $O\Big(I^{h}+\Lambda^{dg}\Big)$\\\hline
        total log-local operations & $\widetilde{O}\Big(k\|H\|_\text{max}t\big(I^{h}+\Lambda^{dg}\big)\Big)$
    \end{tabular}
    \caption{Summary of qubit and gate count costs for our algorithm, for a second-quantized Hamiltonian $H$ in the plane wave momentum basis. The parameters are defined in \cref{tab:glossary}.}
    \label{tab:costs}
\end{table}

\subsection{Comparison to direct encoding}
\label{comparison_direct}

Recall that the goal of the compact encoding is to minimize the number of qubits required to simulate a second-quantized Hamiltonian.
Direct encodings, which explicitly store information about every mode including the unoccupied modes, will require more qubits but afford simpler operations, as discussed in the introduction.

In direct encoding we store the occupation of every mode in a Fock state.
Hence each mode can be assigned to a specific register of qubits, so it is not necessary to store the information identifying the mode in the qubit state.
Therefore, the number of qubits required for a single mode is just $O(\log W)$, where as above $W$ is the single-mode occupation cutoff.
This is multiplied by the number of modes to give the total number of qubits required for the direct encoding:
\begin{equation}
    O\left((\log W)\prod_{j=1}^d\left(\Ucutoff_j-\Lcutoff_j\right)\right),
\end{equation}
where the number of modes is $O\left(\prod_{j=1}^d\left(\Ucutoff_j-\Lcutoff_j\right)\right)$ assuming the numbers of species and non-momentum quantum numbers are constant.
In other words, compared to the number \eqref{qubitn_main} of qubits for the compact encoding, the direct encoding has linear rather than logarithmic scaling with the number of modes.
Hence, when the maximum number $I$ of distinct occupied modes is much smaller than the total number of modes, the compact encoding will be asymptotically advantageous in number of qubits.

The costs of oracle implementations for the direct encoding were evaluated in Section I.B of the Supplemental Material to~\cite{kirby2020sparsevqe}.
As with the compact encoding, the cost of the enumerator oracle dominates, coming out to
\begin{equation}
    O\left((h+g)\log W\right)
\end{equation}
Toffoli gates (using our notation).
As expected, in typical cases this will be smaller than the cost \eqref{oracle_cost} of the enumerator oracle in compact encoding, since it could only be of the same order if the maximum occupation $W$ of a single mode is exponentially larger than $I$, the number of distinct occupied modes, and $\Lambda$, the momentum cutoff.

These comparisons confirm the expected relation between direct and compact encodings: they form a space-time tradeoff, with the direct encoding using more space to obtain shorter circuits, and the compact encoding saving space at the expense of longer circuits.
Note, however, that both are efficient in the sense that their costs in both space and time are at worst polynomial in the problem parameters.
The differences are in which scalings are logarithmic (or constant) versus polynomial.

\subsection{Comparison to prior work on electronic-structure Hamiltonians}
\label{prior_work}

Previous work has demonstrated how to implement sparsity-based simulation of the electronic-structure problem in second-quantization~\cite{babbush16a} and the configuration-interaction (CI) representation~\cite{toloui2013cimatrix,babbush17a}.
These result in gate counts of $\widetilde{O}(N^5t)$ and $\widetilde{O}(\eta^2N^3t)$, respectively, where $N$ is the number of orbitals, $\eta$ is the number of electrons, and $t$ is the simulation time.
The dependence on error is suppressed in these expressions, but is polylogarithmic in the inverse error.

We can compare our method to~\cite{toloui2013cimatrix,babbush16a,babbush17a} by applying it to the electronic-structure problem.
In this case, we can replace $I$ (the maximum possible number of occupied modes) by $\eta$.
The total number of modes in our method is $O(\Lambda^d)$, so we may replace $\Lambda^d$ by $N$ in our asymptotic expressions.
If we apply our method to the CI-matrix, Eq.~(20) in~\cite{babbush17a} gives the sparsity as $O(\eta^2N^2)$, and the discussion following Eq.~(73) in~\cite{babbush17a} shows that $\|H\|_\text{max}$ is polylogarithmic in $N$.
Finally, $g$ and $h$ are both two for the electronic-structure problem.
Making all of these replacements in \eqref{total_cost} and suppressing polylogarithmic factors gives
\begin{equation}
    \widetilde{O}\bigg(\eta^2N^2t(\eta^2+N^2)\bigg)=\widetilde{O}\bigg(\eta^2N^4t\bigg).
\end{equation}
This is better than the scaling for the second-quantized algorithm of~\cite{babbush16a}, but worse by a factor of $N$ than the CI algorithm of~\cite{babbush17a}.
The extra factor of $N$ essentially comes from the fact that the algorithm of~\cite{babbush17a} uses the Slater rules directly, which our algorithm does not take into account.

This is illustrative of what we expect to be a general pattern: while our algorithm is applicable to a broad range of second-quantized Hamiltonians, if special structure is known about some particular Hamiltonian it may be possible to design algorithms that are specific to that Hamiltonian and outperform ours.
An interesting question for future work is to what extent it is possible to design general-purpose algorithms that are able to naturally take advantage of such problem-specific structure.

\section{Compact encoding and Hamiltonians}

In this section, we define the compact encoding, which maps Fock states to qubit states, and the input model for our second-quantized Hamiltonians.
After this section, we will often say ``Fock state" when we really mean ``compact-encoded Fock state," since the latter is cumbersome.
There will usually be no ambiguity in this, since qubit operators can only act upon compact-encoded Fock states, but whenever there is ambiguity we will explicitly state which we are talking about.

\subsection{Compact encoding of Fock states}
\label{encoding}

Throughout, when we refer to momenta we will mean dimensionless momenta, denoted by $\textbf{n}$.
These are related to the dimensionful momenta $\textbf{p}$ as
\begin{equation}
    p_j=\frac{2\pi}{L_j}n_j
\end{equation}
where $L_j$ is the box size for each component $j$ and we take $\hbar=1$.
We also impose cutoffs $\Lambda$ on the momenta, i.e., each component $n_j$ must satisfy
\begin{equation}
\label{cutoff_condition}
    \Lcutoff_j\le n_j\le\Ucutoff_j
\end{equation}
for some cutoffs
\begin{equation}
\label{cutoff_def}
    \Lambda = (~(\Lcutoff_1,\Ucutoff_1),~(\Lcutoff_2,\Ucutoff_2),...,~(\Lcutoff_d,\Ucutoff_d)~),
\end{equation}
where $d$ is the number of spatial dimensions.
In equal-time quantization, it is generally the case that each $\Lcutoff_i<0$ and each $\Ucutoff_i>0$, while in light-front quantization there is some particular axis $z$ such that $\Lcutoff_z>0$ (see \cref{lf_application} for details of light-front quantization).

A Fock state in compact encoding has the form:
\begin{equation}
\label{fock_state}
    |\mathcal{F}\rangle=|(q_1,\textbf{n}_1,w_1),(q_2,\textbf{n}_2,w_2),...,(q_J,\textbf{n}_J,w_J)\rangle,
\end{equation}
where each $w_i$ is the occupancy of the mode $q_i$ with momentum $\textbf{n}_i$~\cite{kreshchuk20a}.
$q_i$ is a collective label that specifies the particle up to its momentum; for example, $q_i$ might determine whether the particle is a boson or a fermion~\footnote{We leave consideration of exotic particle statistics to future work.}, what species of boson or fermion it is (if multiple are present in the theory), and whether it is a particle or an antiparticle, in addition to properties like spin, flavor, color, etc.
We store only occupied modes, so each occupancy $w_i\ge1$.

\begin{example}
Suppose we have a 1+1D theory containing bosons and fermions whose only quantum number is momentum.
We can let $q_i=0$ label bosons, and $q_i=1$ label fermions (and $q_i=2$ label antifermions, but for simplicity we will not include these in the examples below).
Then a few examples of compact-encoded Fock states are:
\begin{equation}
    |(0,2,1)\rangle,
\end{equation}
which encodes one boson with momentum $2$.
\begin{equation}
    |(0,2,2)\rangle
\end{equation}
encodes two bosons with momentum $2$.
\begin{equation}
    |(0,2,3),(1,5,1)\rangle
\end{equation}
encodes three bosons with momentum $2$ and one fermion with momentum $5$.
\begin{equation}
    |(0,2,3),(0,3,2),(1,5,1)\rangle
\end{equation}
encodes three bosons with momentum $2$, two bosons with momentum $3$, and one fermion with momentum $5$.
If our theory also contained spin, for example, then we would expand the $q_i$ labels to include this: e.g., $q_i=\{1,\uparrow\}$ means fermion with spin up, so
\begin{equation}
    \big|(\{1,\uparrow\},2,1),~(\{1,\downarrow\},2,1)\big\rangle
\end{equation}
encodes one fermion with momentum 2 and spin up, and one fermion with momentum 2 and spin down.
\end{example}

We compact-encode a Fock state \eqref{fock_state} in a qubit register of the form
\begin{equation}
\label{qubit_state_encoding}
    |X_1,X_2,...,X_I\rangle,
\end{equation}
where $I$ is the maximum possible number of occupied modes, and each $X_i$ is a \emph{mode register} capable of encoding a single mode $(q_i,\textbf{n}_i,w_i)$.
For a Fock state containing $J\le I$ occupied modes we use the first $J$ of the $X_i$ to encode the modes.
The encoded modes are ordered primarily by $q_i$, and secondarily by momentum.
Note that in equal-time quantization the actual number of occupied modes can be unbounded, so we would have to impose a cutoff $I$ by hand.
In light-front quantization $I$ is finite and determined by the harmonic resolution~\cite{kreshchuk20a}.
In chemistry, the particle number (i.e., the total occupation) is generally fixed, so $I$ is equal to the particle number.

Given some maximum number $I$ of occupied modes, either fixed by the theory or imposed by hand, we have to encode $I$ mode registers.
Each mode register must encode occupation of the mode (which we take to be upper bounded by some cutoff $W$), the non-momentum quantum numbers of the mode (which we assume to take a constant number of possible values, and hence to require some fixed number $N_q$ of qubits), and the momentum of the mode.
For $d$ spatial dimensions indexed by $j$, each component of momentum takes ${\Ucutoff_j-\Lcutoff_j+1}$ values, so the momentum of a mode can be encoded in $\sum_{j=1}^d\left\lceil\log_2\left(\Ucutoff_j-\Lcutoff_j+1\right)\right\rceil$ qubits.
Hence the total number of qubits required to encode a mode is
\begin{equation}
    N_q+\lceil\log_2W\rceil+\sum_{j=1}^d\left\lceil\log_2\left(\Ucutoff_j-\Lcutoff_j+1\right)\right\rceil,
\label{moden}
\end{equation}
so the total number of qubits to encode a Fock state that contains at most $I$ occupied modes is
\begin{equation}
\begin{split}
    \qubitn=I\bigg(&N_q+\lceil\log_2W\rceil\\
    &+\sum_{j=1}^d\left\lceil\log_2\left(\Ucutoff_j-\Lcutoff_j+1\right)\right\rceil\bigg).
\end{split}
\label{qubitn}
\end{equation}

If for some dimension $j$, $\Ucutoff_j$ and $\Lcutoff_j$ are both positive (or both negative), the number of occupied modes $I$ is bounded by some fraction of the total momentum in that dimension.
If this is only true for a single dimension (i.e., all other dimensions can have both positive and negative momentum), then this case corresponds to light-front quantization~\cite{kreshchuk20a}.
Without loss of generality, suppose $\Lcutoff_1>0$.
Let $K$ denote the total momentum in dimension 1.
In this case the maximum total number $I$ of occupied modes in any Fock state is
\begin{equation}
\label{monotonic_number_of_modes}
    I=\left\lfloor\frac{K}{\Lcutoff_1}\right\rfloor,
\end{equation}
since every particle must have momentum at least $\Lcutoff_1$ in dimension 1.
The Fock state satisfying \eqref{monotonic_number_of_modes} is one containing $\lfloor K/\Lcutoff_1\rfloor$ modes, all with momentum $\Lcutoff_1$ or $\Lcutoff_1+1$ in dimension 1 such that the total momentum in dimension 1 is $K$, but with distinct other quantum numbers~\cite{kreshchuk20a}.
In the special case where $\Lcutoff_1=1$, \eqref{monotonic_number_of_modes} simplifies to show that the maximum number of occupied modes is identical to the total momentum along axis 1.

\subsection{Special Case: Light-Front Quantization}
\label{lf_application}

Although the compact encoding is agnostic to the form of the theory to which it is applied, it turns out to be particularly advantageous for relativistic field theories in the light-front (LF) formulation~\cite{PhysRevD.32.1993,PhysRevD.36.1141,brodsky98a,kreshchuk20a,kreshchuk2020light}.
Here, we review light-front quantization and explain how compact encoding applies in this case.
We will later return to the light-front example to illustrate the methods.

We can think of LF quantization as taking the perspective of a massless observer moving at the speed of light in some direction, which we take to be the $-z$ direction.
Thus the dimensionless discretized momenta along this axis take strictly positive values ${n_z\in[1,K]}$, where $K$ is the total dimensionless LF momentum of the Fock state (also called the \textit{harmonic resolution}).
Importantly, this is also true for massless particles~\cite{brodsky98a}.
In other words,
\begin{equation}
\label{cutoff_lf}
    \Lcutoff_1=1,\quad\Ucutoff_1=K,
\end{equation}
where we take axis 1 to correspond to the $z$ direction.
Also, since it is the total LF momentum, $K$ automatically imposes a cutoff on the number of excitations in a mode (${W = K}$), as well as on the number of occupied modes in a Fock state (${I = \sqrt{2K}}$ for ${d=1}$ and ${I = K}$ for $d\ge2$)~\cite{kreshchuk20a}.
The momenta along axes transverse to the light-front direction have the same properties as in equal-time quantization.

The above points mean that for light-front quantization the general expression \eqref{qubitn} for qubit count in compact encoding specializes to
\begin{equation}
\label{qubitn_lf_1D}
    \qubitn_{\text{compact}}^{\text{LF,$d=1$}}=\sqrt{2K}\left(N_q+2\lceil\log_2K\rceil\right)
\end{equation}
for one spatial dimension, or
\begin{equation}
\begin{split}
    \qubitn_{\text{compact}}^{\text{LF,$d\geq2$}}&=K\bigg(N_q+2\lceil\log_2K\rceil\\
    &+\sum_{j=2}^d\left\lceil\log_2\left(\Ucutoff_j-\Lcutoff_j+1\right)\right\rceil\bigg)
\end{split}
\label{qubitn_lf_moreD}
\end{equation}
for $d>1$ spatial dimensions.

We can compare this to the number of qubits required for the direct encoding, as a special case of the general comparison between the two encodings given above.
The total number of modes is
\begin{equation}
\label{number_of_modes_direct}
    K\left(\prod_{j=2}^d\big(\Ucutoff_j-\Lcutoff_j+1\big)\right)q,
\end{equation}
where $q=O(2^{N_q})$ is the number of possible values of the intrinsic quantum numbers, the number of possible values of the light-front momentum is $K$, and the number of possible values of the transverse momenta is $\prod_{j=2}^d\big(\Ucutoff_j-\Lcutoff_j+1\big)$.
In a direct encoding we would encode the occupancy of each of these modes in $\lceil\log_2K\rceil$ qubits (since $K$ is an upper bound on the occupancy), so the total number of qubits for the direct encoding is
\begin{equation}
\label{number_of_qubits_direct}
    \qubitn_\text{direct}^{\text{LF}}=K\left(\prod_{j=2}^d\big(\Ucutoff_j-\Lcutoff_j+1\big)\right)q\,\lceil\log_2
    K\rceil.
\end{equation}

In other words, for $\Lambda^\perp$ an upper bound on the transverse momentum cutoffs, up to constant and logarithmic factors the number of qubits for the direct encoding is $\widetilde{\Theta}\left(K(\Lambda^\perp)^{d-1}\right)$, while the number of qubits for the compact encoding is $\widetilde{\Theta}\left(K\right)$.
This explains why LF quantization motivates development of compact encoding methods.
In \cref{applications}, we will analyze our oracle constructions for a number of field theories in both equal-time and light-front quantization.

\subsection{Hamiltonian}
\label{hamiltonian}

A normal-ordered, second-quantized Hamiltonian is composed of terms with the form
\begin{equation}
\label{2_quant_term}
    \beta a^\dagger_ia^\dagger_j\cdots a^\dagger_ka_la_m\cdots a_n,
\end{equation}
where $\beta$ is a coefficient, $a^\dagger$ and $a$ are fermionic or bosonic creation and annihilation operators, and $i,j,...,k,l,m,...,n$ are labels for the particles being created and annihilated.
In the remainder of this paper, we will write creation and annihilation operators as
\begin{equation}
    a^{(\dagger)}_{q_i}(\textbf{n}_i),
\end{equation}
where $\textbf{n}_i$ is the momentum of the created or annihilated particle and $q_i$ is a collective label for its remaining quantum numbers (including species), as in the previous section.

We may visualize a term like \eqref{2_quant_term} as a diagram with an incoming line for each annihilation operator and an outgoing line for each creation operator.
We define an \emph{interaction} to be a sum of such terms, with the momenta varying over all momentum-conserving combinations, but the other properties of the incoming and outgoing particles fixed.
Thus we can visualize an interaction as a diagram without momentum specifications.
An interaction whose diagram contains $f$ external lines is called an \emph{$f$-point interaction}.
In the remainder of the paper, when we refer to incoming or outgoing particles, we will mean the incoming or outgoing lines of the diagram of an interaction, which represent annihilation or creation operators, respectively.
Note that although a diagram of this kind resembles a Feynman diagram, it does not represent a matrix element calculation, but instead is just a visualization of a collection of ladder operator monomials.

\begin{example}
\label{interaction_example}
In a 1+1D theory consider the 3-point interaction
\begin{equation}
\label{example_interaction}
    \sum_{\textbf{n}_1,\textbf{n}_2,\textbf{n}_3}a_2^\dagger(\textbf{n}_3)a_1(\textbf{n}_2)a_1(\textbf{n}_1),
\end{equation}
where the sum varies over all momenta such that $\textbf{n}_1+\textbf{n}_2=\textbf{n}_3$.
This interaction describes annihilation of a pair of particles of type `1' to form a single particle of type `2', and is represented by the diagram shown in \cref{feynman_diag_example}.
So, for example, one possible instance of the interaction would map two particles of type `1', both with momentum 2 (i.e., $\textbf{n}_1=\textbf{n}_2=2$), to a particle of type `2' with momentum  (i.e., $\textbf{n}_3=4$).
We can represent these as Fock states in compact encoding as in \eqref{fock_state}:
\begin{equation}
    |(\text{`1'},2,2)\rangle~\rightarrow~|(\text{`2'},4,1)\rangle,
\end{equation}
where we recall that the first entry in each tuple encodes $q_i$ (in this case `1' or `2'), the second entry encodes the momentum, and the third entry encodes the occupation.
If instead the incoming momenta were $\textbf{n}_1=1$ and $\textbf{n}_2=3$, then the incoming and outgoing Fock states would instead be represented as
\begin{equation}
    |(\text{`1'},1,1),(\text{`1'},3,1)\rangle~\rightarrow~|(\text{`2'},4,1)\rangle.
\end{equation}
If another, non-interacting mode were present (say, two particles of type `2' with momentum 5), then we would have
\begin{equation}
    |(\text{`1'},1,1),(\text{`1'},3,1),(\text{`2'},5,2)\rangle\rightarrow|(\text{`2'},4,1),(\text{`2'},5,2)\rangle,
\end{equation}
where the additional mode $(\text{`2'},5,2)$ representing the non-interacting particles is present on both sides.
Note that since all momenta are positive in the above examples, they represent possible interaction instances in a 1+1D light-front field theory.
\end{example}

\begin{figure}
\centering
\includegraphics[width=2in]{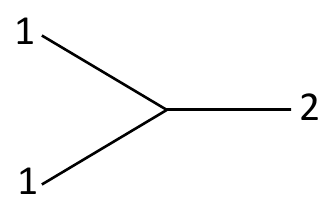}
\caption{Diagram of the example 3-point interaction \eqref{example_interaction}, which describes annihilation of a pair of particles of type `1' to form a single particle of type `2' (we read time from left to right).
\label{feynman_diag_example}}
\end{figure}

We can formally define an interaction as follows:
\begin{definition}
\label{interaction_def}
    An \emph{interaction} is specified by the set
    \begin{equation}
    \label{interaction_eq}
        \{(q_1,q_2,...,q_g),(q_{g+1},q_{g+2},...,q_f)\},
    \end{equation}
    together with a coefficient function $\beta$ that maps sets of momenta for the incoming and outgoing particles to coefficient values.
    The $q_1,q_2,...,q_g$ specify the outgoing particles, and $q_{g+1},q_{g+1},...,q_f$ specify the incoming particles, up to their momenta.
\end{definition}
For an interaction as in \cref{interaction_def}, the corresponding \emph{interaction Hamiltonian} is
\begin{equation}
\label{interaction_hamiltonian}
    H_I\equiv\sum_{\{\textbf{n}_i\}}\beta(\{\textbf{n}_i\})\left(\prod_{i=1}^ga^\dagger_{q_i}(\textbf{n}_i)\right)\left(\prod_{i=g+1}^fa_{q_i}(\textbf{n}_i)\right),
\end{equation}
where the sum runs over all sets $\{\textbf{n}_i\}$ that conserve total momentum, i.e., such that
\begin{equation}
\label{momentum_cons}
    \sum_{i=1}^g\textbf{n}_i=\sum_{i=g+1}^f\textbf{n}_i.
\end{equation}
Any second-quantized Hamiltonian may be expressed as a sum of interaction Hamiltonians $H_I$ of the form \eqref{interaction_hamiltonian}.

Notice that since the Hamiltonian must be Hermitian, for each interaction the Hamiltonian must also contain its Hermitian conjugate.
For example, a Hamiltonian containing the interaction in \cref{interaction_example} (Eq.~\eqref{example_interaction}) must also contain the interaction
\begin{equation}
\label{example_interaction_reverse}
    \sum_{\textbf{n}_1,\textbf{n}_2,\textbf{n}_3}a_1^\dagger(\textbf{n}_1)a_1^\dagger(\textbf{n}_2)a_2(\textbf{n}_3),
\end{equation}
where again the sum varies over all momenta such that $\textbf{n}_1+\textbf{n}_2=\textbf{n}_3$.

In this paper, we will assume that the Hamiltonian, and thus the interactions included in it, are fixed up to momentum cutoffs.
Each interaction sets particular values of $f$ and $g$, so $f$ and $g$ can be treated as constants.
We will focus on the scaling of our algorithms with momentum cutoffs, which specify the resolution at which we study the given Hamiltonian.

\section{Sparse Hamiltonians}
\label{sparsity}

A Hamiltonian written as a matrix in a particular basis is said to be \emph{sparse} if the number of nonzero elements in each row (or column) is polylogarithmic in the total Hilbert space dimension.
Similarly, the maximum number of nonzero elements in any row (or column) is called the \emph{sparsity} of the Hamiltonian.
In this section, we first review methods for simulating sparse Hamiltonians, then describe how we can use these methods to act on Fock states in the compact encoding as described in \cref{encoding}, and finally analyze the sparsity of interaction Hamiltonians of the form \eqref{interaction_hamiltonian}.

\subsection{Sparse Hamiltonian Simulation Review}
\label{sparse_review}

Aharonov and Ta-Shma presented the first quantum algorithm for simulating sparse Hamiltonians in 2003~\cite{aharonov03a}, while the same year Childs \emph{et al.}~demonstrated quantum advantage with respect to an oracle in a similar setting~\cite{childs03a}.
Subsequent works have extended and improved these results~\cite{berry07a,childs10a,berry12a,berry14a,berry15a,berry15b,low17a,low19a,berry20a}.
These methods are based on accessing the sparse Hamiltonians via oracle input models.

Early results in sparse Hamiltonian simulation were based on product formulas~\cite{aharonov03a,berry07a,berry14a,berry15a} or quantum walks~\cite{childs10a,berry12a}.
The product formula based methods ultimately achieved optimal dependence on the error $\eps$ of the simulation~\cite{berry14a,berry15a}, while the quantum walk based methods achieved optimal dependence on the sparsity and simulation time~\cite{berry12a} (these optimal scalings are discussed below).
Then, Berry \emph{et al.}~used a quantum walk structure with techniques borrowed from the product formula approaches to obtain near optimal dependence on all parameters~\cite{berry15b}, and a subsequent paper extended and improved these results for time-dependent Hamiltonians~\cite{berry20a}.
In specific cases, related methods that depend only on the interaction-picture or off-diagonal norms of the Hamiltonian may be advantageous, but may also require extra work to cast the Hamiltonian from the sparse oracle input model to the required forms~\cite{low18a,kalev2021offdiagonal,chen2021permutationexpansion}.
Finally, Low and Chuang developed a technique based on quantum signal processing called \emph{qubitization} that achieved fully optimal scaling with all parameters for the time-independent case~\cite{low17a,low19a}.

Recent works on sparse Hamiltonian simulation specify the Hamiltonian via a pair of oracles, which may be expressed as unitary operations.
The first oracle is typically defined in the quantum walk based approaches as follows:
\begin{equation}
\label{enumerator_oracle_standard}
    O_F|x,i\rangle=|x,y_i\rangle,
\end{equation}
where for a Hamiltonian $H$ with sparsity $k$, ${i\in[k]\equiv\{0,1,2,...,k-1\}}$, and $y_i$ is the index of the $i$th nonzero entry in row $x$ of $H$.
We refer to $O_F$ as the ``enumerator oracle."
The basic quantum walk step developed in~\cite{childs10a,berry12a} underlies the near-optimal algorithms of~\cite{berry15b,berry20a} as well as the optimal algorithm obtained by qubitization~\cite{low19a}, so these algorithms use the form of the enumerator oracle given in \eqref{enumerator_oracle_standard}.

The product formula based methods~\cite{aharonov03a,berry07a,berry14a,berry15a}, on the other hand, typically define the enumerator oracle as follows:
\begin{equation}
\label{enumerator_oracle_standard_alt}
    O'_F|x,i,0\rangle=|x,i,y_i\rangle,
\end{equation}
i.e., the index $i$ is saved rather than being uncomputed on the way to computing $y_i$.
This distinction between uncomputing or saving the index $i$ appears harmless, but we will see shortly that for the variant of the oracles we will require, some care is needed to properly employ $O_F$ in order to obtain the optimal scaling offered by qubitization~\cite{low19a} or the near-optimal algorithm for time-dependent Hamiltonians~\cite{berry20a}.

The second oracle (which is common to all of the methods) calculates matrix elements of the Hamiltonian given indices for entries in the Hamiltonian:
\begin{equation}
\label{matrix_element_oracle_standard}
    O_H|x,y,0\rangle=|x,y,H_{xy}\rangle.
\end{equation}
The 0 on the left-hand side above denotes a register containing the number of qubits necessary to store the value of the entry $H_{xy}$ in binary form with the desired precision.
Note that although $O_H$ is defined for arbitrary matrix elements, it is typically only applied to pairs of indices $(x,y_i)$ corresponding to nonzero matrix elements (i.e., generated by the enumerator oracle), and this will always be the case for us, which simplifies the construction (see \cref{matrix_element}, and the proof of \cref{construction_lemma} in the next subsection).
We refer to $O_H$ as the ``matrix element oracle."

To simulate evolution for time $t$ under a time-independent Hamiltonian $H$ with sparsity $k$, qubitization uses
\begin{equation}
\label{query_complexity}
    O\left(\tau+\frac{\log(1/\eps)}{\log\log(1/\eps)}\right)
\end{equation}
oracle queries~\cite[Corollary 15]{low19a}, where
\begin{equation}
    \tau\equiv k\|H\|_\text{max}t\quad\text{(time-independent $H$)}
\end{equation}
and $\|H\|_\text{max}$ (the max-norm of $H$) is defined to be the maximum magnitude of any entry in $H$.
This scaling is optimal in the error $\eps$~\cite[Theorem 2.2]{berry14a}, and in the simulation time $t$ and sparsity $k$~\cite{berry12a}.
The optimal scaling with the simulation time is set by the \emph{no fast-forwarding theorem}, which states that evolution under a general Hamiltonian for a time $t$ cannot be simulated using a number of operations that is sublinear in $t$~\cite[Theorem 3]{berry07a} (this can be violated for some special types of Hamiltonians~\cite{atia17a}).

To simulate evolution for a time $T$ under a time-dependent Hamiltonian $H(t)$, the method of~\cite{berry20a} requires
\begin{equation}
\label{query_complexity_time_dep}
    O\left(\tau\frac{\log(\tau/\eps)}{\log\log(\tau/\eps)}\right)
\end{equation}
oracle queries (note that \eqref{query_complexity_time_dep} is a product, whereas \eqref{query_complexity} is a sum), with $\tau$ now defined by
\begin{equation}
    \tau\equiv k\int_{0}^{T}\|H(t)\|_\text{max}\,dt\quad\text{(time-dependent $H$)},
\end{equation}
where without loss of generality we take the initial time to be $t=0$.
In other words, the dependence on $H(t)$ obtained in~\cite{berry20a} is given by the $L^1$-norm of $\|H(t)\|_\text{max}$ over the time interval $[0,T]$.
This satisfies the intuitive notion that the cost of simulating $H(t)$ should depend instantaneously only on the value of $H$ at the current time.
For comparison, previous works on the time-independent case can generalize to time-dependent Hamiltonians, but with query complexity that scales instead according to $\text{max}_{t\in[0,T]}\|H(t)\|_\text{max}$ (i.e., with the $L^\infty$-norm of $\|H(t)\|_\text{max}$ over the time interval $[0,T]$)~\cite{berry20a}.

The method of~\cite{berry20a} uses a rescaling of the Hamiltonian depending on its instantaneous max-norm, which is accessed via two additional oracles:
\begin{equation}
\label{extra_oracles_for_time_dep}
\begin{split}
    &O_\text{norm}|t,z\rangle=|t,z\oplus\|H(t)\|_\text{max}\rangle,\\
    &O_\text{var}|\sigma,z\rangle=|\sigma,z\oplus f^{-1}(\sigma)\rangle,
\end{split}
\end{equation}
where
\begin{equation}
    f(t)\equiv\int_0^T\Lambda(t)\,dt
\end{equation}
and $\Lambda(t)$ is any efficiently computable tight upper bound on $\|H(t)\|$ (see~\cite[Section 4]{berry20a}).
As noted in~\cite[Section 4.2]{berry20a}, $f^{-1}$ can be computed to precision $\delta$ using $O(\log(T/\delta)$ queries to $f$, so as long as $\|H\|_\text{max}(t)$ is efficiently computable for any $t$, we can efficiently implement the oracles \eqref{extra_oracles_for_time_dep}.

These best known methods for simulation of sparse Hamiltonians have in common the basic step that they use to access the Hamiltonian.
This step is implementation of an isometry typically labeled $T$, which was originally proposed in~\cite{childs10a} as a component of a quantum walk, and first used explicitly in a Hamiltonian simulation technique in~\cite{berry12a}.
In our notation, $T$ may be written:
\begin{equation}
\label{Tdef}
    T=\sum_{x=1}^{2^n}|x\rangle_a|\phi_x\rangle_{b,c}\langle x|_a,
\end{equation}
where
\begin{equation}
\label{phidef_old}
\begin{split}
    &|\phi_x\rangle_{b,c}\\
    &\equiv\sqrt{\frac{r}{\|H\|_1}}\sum_{y=1}^{2^n}\sqrt{H_{xy}^*}|y\rangle_b|0\rangle_c+\sqrt{1-\frac{r\sigma_x}{\|H\|_1}}|\zeta_x\rangle_b|1\rangle_c,
\end{split}
\end{equation}
and $a$ and $b$ label registers of the same number of qubits, and $c$ labels a single ancilla qubit (we will typically suppress these subscripts when doing so leads to no ambiguity).
Here $r\in(0,1]$ is a parameter,
\begin{equation}
    \sigma_x\equiv\sum_{y=1}^{2^n}|H_{xy}|,
\end{equation}
and $|\zeta_x\rangle$ is some linear combination of the $|y\rangle$.
A careful reader may note that $T$ as defined in \eqref{Tdef} is not unitary.
This is resolved by letting $\eqref{Tdef}$ only define the action on a state of the form $|x\rangle|0\rangle$, i.e.,
\begin{equation}
    |x\rangle_a|0\rangle_{b,c}~\xrightarrow{T}~|x\rangle_a|\phi_x\rangle_{b,c},
\end{equation}
and the action on states not of this form can be anything as long as the overall operation is unitary~\cite{berry12a}.

The final, single ancilla qubit in \eqref{phidef_old} (i.e., the qubit labeled $c$) is present in order to ensure that the last term (proportional to $|\zeta_x\rangle|1\rangle$) is orthogonal to any of the first set of terms for any $x$, which is required by Eq.~(25) in~\cite{childs10a}.
Note that the final terms need not be orthogonal to each other for distinct values of $x$, even though a superficial reading of~\cite{childs10a} might suggest otherwise.
In fact, it is the final term in $|x\rangle|\phi_x\rangle$ (rather than just $|\phi_x\rangle$) that corresponds to the state $|\hspace{-4.5pt}\perp_j\rangle$ in~\cite{childs10a}: this term must therefore take orthogonal values for distinct values of $x$ (see Eq.~(24) and the corresponding discussion in~\cite{childs10a}), but this is trivially satisfied, since the final term in $|x\rangle|\phi_x\rangle$ is proportional to $|x\rangle|\zeta_x\rangle|1\rangle$.

A complete description of the various ways that the isometry $T$ can be used to construct Hamiltonian simulation algorithms is beyond the scope of this paper, but it was originally introduced in order to construct the quantum walk operator
\begin{equation}
    U=iS(2TT^\dagger-1)
\end{equation}
in~\cite{childs10a}, where $S$ is the operator that swaps registers $a$ and $b$, and also swaps the ancilla qubit $c$ with an additional ancilla qubit initially in the $|0\rangle$ state.
Repeatedly applying the quantum walk step $U$ yields a discrete approximation of the Hamiltonian evolution, up to unitary equivalence~\cite{childs10a,berry12a}.


\subsection{Sparse Hamiltonian Simulation for Compact-Encoded Fock States}

In the sparse simulation methods above, the oracles act upon states that encode row and column indices of the Hamiltonian.
We instead want to use oracles that act upon Fock states (recall that throughout, when we say ``Fock states" we mean ``compact-encoded Fock states").
The best sparse Hamiltonian simulation methods access the Hamiltonian via the operator $T$ (defined in \eqref{Tdef}), as discussed in \cref{sparse_review}.
Therefore, we want to use our new oracles to implement a version of $T$ that acts on Fock states rather than on row indices.
This will allow the sparse simulation methods to be implemented directly on Fock states.

One complication is the fact that for a $k$-sparse Hamiltonian, there are at most $k$ nonzero entries in any row or column, but in general there can be fewer than $k$ in some rows and columns.
In such a row, some of the values of $i$, which runs from $0$ to $k-1$ and is supposed to index the nonzero entries in a given row, cannot actually index nonzero entries, because there aren't enough nonzero entries in the row.
There is no obvious natural mapping from these unused values of $i$ to matrix entries.
We will see that this situation arises very commonly for interaction Hamiltonians.
We will define oracles that act on Fock states in a way that resolves this issue, and show that we can use these to recover the desired building blocks for the sparse simulation algorithms described above.
In \cref{enumerator,matrix_element}, we will then explicitly construct implementations of these oracles.

Let $|\mathcal{F}\rangle$ be a Fock state, and let $H$ be a $k$-sparse Hamiltonian.
Then there are at most $k$ states $|\mathcal{F}'_i\rangle$ whose Hamiltonian matrix elements with $|\mathcal{F}\rangle$ are nonzero: assume that they are indexed by elements of some set $\mathcal{I}_\mathcal{F}\subseteq[k]$.
Let $O_F$ and $O_H$ now define oracles that act as follows:
\begin{align}
    &O_F|\mathcal{F}\rangle|0\rangle|i\rangle=|\mathcal{F}\rangle|\Psi'(\mathcal{F},i)\rangle|a(\mathcal{F},i)\rangle,\label{enumerator_oracle_def}\\
    &O_H|\mathcal{F}\rangle|\mathcal{F}'\rangle=|\mathcal{F}\rangle|\mathcal{F}'\rangle\big|\langle\mathcal{F}'|H|\mathcal{F}\rangle\big\rangle,\label{matrix_element_oracle_def}
\end{align}
where $i\in[k]$ and the functions $\Psi'(\mathcal{F},i)$ and $a(\mathcal{F},i)$ are defined by
\begin{equation}
    \label{bigpsidef}
    \Psi'(\mathcal{F},i)=
    \begin{cases}
        \mathcal{F}'_i~\text{if}~i\in\mathcal{I}_\mathcal{F},\\
        \mathcal{F}~\text{if}~i\notin\mathcal{I}_\mathcal{F},
    \end{cases}
\end{equation}
\begin{equation}
\label{afuncdef}
    a(\mathcal{F},i)=
    \begin{cases}
        0~\text{if}~i\in\mathcal{I}_\mathcal{F},\\
        i~\text{if}~i\notin\mathcal{I}_\mathcal{F}.
    \end{cases}
\end{equation}
Without loss of generality, let $\mathcal{F}'_0=\mathcal{F}$ whenever the matrix element of $\mathcal{F}$ with itself is nonzero.
Thus the enumerator oracle $O_F$ may be alternatively expressed as
\begin{equation}
\label{enumerator_oracle_def_alt}
    O_F|\mathcal{F}\rangle|0\rangle|i\rangle=
    \begin{cases}
        |\mathcal{F}\rangle|\mathcal{F}'_i\rangle|0\rangle~\text{if}~i\in\mathcal{I}_\mathcal{F},\\
        |\mathcal{F}\rangle|\mathcal{F}\rangle|i\rangle~\text{if}~i\notin\mathcal{I}_\mathcal{F}.
    \end{cases}
\end{equation}

Next, we define an analog of $|\phi_x\rangle$ \eqref{phidef_old} in terms of Fock states:
\begin{equation}
\begin{split}
    |\phi_\mathcal{F}\rangle\equiv&\sqrt{\frac{r}{\|H\|_1}}\sum_{|\mathcal{F}'\rangle}\sqrt{\langle\mathcal{F}'|H|\mathcal{F}\rangle}|\mathcal{F}'\rangle|a(\mathcal{F},\mathcal{F}')\rangle|0\rangle\\
    &+\sqrt{1-\frac{r\sigma_\mathcal{F}}{\|H\|_1}}|\zeta_\mathcal{F}\rangle|1\rangle,
\label{phidef}
\end{split}
\end{equation}
where $a(\mathcal{F},\mathcal{F}')$ is some binary number that is zero when $\langle\mathcal{F}'|H|\mathcal{F}\rangle\neq0$, $\sigma_\mathcal{F}$ is defined analogously to $\sigma_x$, i.e.,
\begin{equation}
    \sigma_\mathcal{F}\equiv\sum_{|\mathcal{F}''\rangle}|\langle\mathcal{F}|H|\mathcal{F}''\rangle|,
\end{equation}
and $|\zeta_\mathcal{F}\rangle$ is a linear combination of states of the form
\begin{equation}
\label{zeta_F_def}
    |\mathcal{F}''\rangle|b\rangle,
\end{equation}
where $|b\rangle$ is some binary number encoded in the same register as $a(\mathcal{F},\mathcal{F}')$.
All of these components will be determined precisely by the algorithm for constructing $|\phi_\mathcal{F}\rangle$, below.
The fact that $a(\mathcal{F},\mathcal{F}')=0$ whenever $\langle\mathcal{F}'|H|\mathcal{F}\rangle\neq0$ ensures that for
\begin{equation}
    Q(\mathcal{F})\equiv\{|\mathcal{F}'\rangle~|~\langle\mathcal{F}'|H|\mathcal{F}\rangle\neq0\},
\end{equation}
we may rewrite \eqref{phidef} as
\begin{equation}
\begin{split}
    |\phi_\mathcal{F}\rangle=&\sqrt{\frac{r}{\|H\|_1}}\sum_{|\mathcal{F}'\rangle
    \in Q(\mathcal{F})}\sqrt{\langle\mathcal{F}'|H|\mathcal{F}\rangle}|\mathcal{F}'\rangle|0\rangle|0\rangle\\
    &+\sqrt{1-\frac{r\sigma_\mathcal{F}}{\|H\|_1}}|\zeta_\mathcal{F}\rangle|1\rangle.
\label{phidef_alt}
\end{split}
\end{equation}
Using $|\phi_\mathcal{F}\rangle$, we define a version of $T$ for Fock states:
\begin{equation}
\label{newTdef}
    T\equiv\sum_{|\mathcal{F}\rangle}|\mathcal{F}\rangle|\phi_\mathcal{F}\rangle\langle\mathcal{F}|.
\end{equation}
In \cref{construction_lemma}, below, we show how $T$ can be constructed using $O(1)$ queries to the oracles $O_F$ and $O_H$ as defined by \eqref{enumerator_oracle_def} and \eqref{matrix_element_oracle_def}.
First, however, we will show that $T$ as defined by \eqref{newTdef} may replace the original version of $T$ in \eqref{Tdef}, and the resulting operator acts on Fock states but otherwise reproduces all of the properties of the basic step used in~\cite{berry12a}.

The terms in the first line of \eqref{phidef_alt} are exact analogs of the corresponding terms in \eqref{phidef_old}, the definition of $|\phi_x\rangle$ used to construct the standard operator $T$ as used in~\cite{berry12a,berry15b,low19a,berry20a}.
The only difference is the inclusion of an additional ancilla register in state $|0\rangle$ instead of just the single ancilla qubit in state $|0\rangle$.

The second line in \eqref{phidef_alt} is also analogous to the corresponding term in \eqref{phidef_old}, but here $|\zeta_\mathcal{F}\rangle$ includes the ancilla register $|b\rangle$ (see \eqref{zeta_F_def}) in addition to the register encoding Fock states, whereas $|\zeta_x\rangle$ in \eqref{phidef_old} is a linear combination of row indices only.
However, in order to satisfy the orthogonality conditions discussed in the final paragraph of \cref{sparse_review}, we only require that for any $\mathcal{F}$, $\mathcal{F}'$, and $\mathcal{F}''$, $|\zeta_\mathcal{F}\rangle|1\rangle$ is orthogonal to $|\mathcal{F}''\rangle|a(\mathcal{F}',\mathcal{F}'')\rangle|0\rangle$; this is satisfied because of the final single qubit.
These are the only conditions that $T$ must satisfy~\cite{berry12a}.

\begin{lemma}[Construction of $T$]
\label{construction_lemma}
    Given an input state $|\mathcal{F}\rangle$, the operator $T$ as defined in \eqref{newTdef} can be implemented using $O(1)$ queries to the oracles $O_F$ and $O_H$ as defined in \eqref{enumerator_oracle_def} and \eqref{matrix_element_oracle_def}.
    (This Lemma closely follows Lemma 4 in~\cite{berry12a}.)
\end{lemma}
\begin{proof}
The operator $T$ maps a Fock state $|\mathcal{F}\rangle$ to $|\mathcal{F}\rangle|\phi_\mathcal{F}\rangle$, for $|\phi_\mathcal{F}\rangle$ as defined by \eqref{phidef_alt}.
Explicitly including ancillas, we assume an input state of the form
\begin{equation}
    |\mathcal{F}\rangle|0\rangle|0\rangle|0\rangle.
\end{equation}
We map this to $|\mathcal{F}\rangle|\phi_\mathcal{F}\rangle$ as follows:
\begin{enumerate}
    
    \item Prepare a uniform superposition of the indices $i=1,2,...,k$ as:
    \begin{equation}
    \label{step1}
        \frac{1}{\sqrt{k}}\sum_{i=1}^k|\mathcal{F}\rangle|0\rangle|i\rangle|0\rangle.
    \end{equation}
    
    \item Apply $O_F$ as defined in \eqref{enumerator_oracle_def} to the first three registers, obtaining
    \begin{equation}
    \label{step2}
        \frac{1}{\sqrt{k}}\sum_{i=1}^k|\mathcal{F}\rangle|\Psi'(\mathcal{F},i)\rangle|a(\mathcal{F},i)\rangle|0\rangle.
    \end{equation}
    
    \item Controlled on $a(\mathcal{F},i)=0$, apply $O_H$ as defined in \eqref{matrix_element_oracle_def} to the first two registers, to calculate $\langle\mathcal{F}'_i|H|\mathcal{F}\rangle$ in an ancilla register that is initially $|0\rangle$ (recall that $a(\mathcal{F},i)=0$ if and only if $\langle\mathcal{F}'_i|H|\mathcal{F}\rangle\neq0$).
    Then, controlled on the resulting value $\langle\mathcal{F}'_i|H|\mathcal{F}\rangle$, rotate the single ancilla qubit (the final register) as
    \begin{equation}
    \label{singlequbitrotation}
        |0\rangle~\mapsto~\sqrt{k\frac{r\langle \mathcal{F}'_i|H|\mathcal{F}\rangle}{\|H\|_1}}|0\rangle+\sqrt{1-k\frac{r|\langle \mathcal{F}'_i|H|\mathcal{F}\rangle|}{\|H\|_1}}|1\rangle.
    \end{equation}
    Finally, uncompute the ancilla register encoding $\langle\mathcal{F}'_i|H|\mathcal{F}\rangle$ using another controlled query to $O_H$.
    This step is identical to the corresponding step in the method described in the proof of Lemma 4 in~\cite{berry12a}.

\end{enumerate}

To obtain the full state after these steps are complete, we insert \eqref{singlequbitrotation} into \eqref{step2}, giving
\begin{equation}
\begin{split}
    &\sum_{i=1}^k\sqrt{\frac{r\langle \mathcal{F}'_i|H|\mathcal{F}\rangle}{\|H\|_1}}|\mathcal{F}\rangle|\Psi'(\mathcal{F},i)\rangle|a(\mathcal{F},i)\rangle|0\rangle\\
    &+\sum_{i=1}^k\sqrt{\frac{1}{k}-\frac{r|\langle \mathcal{F}'_i|H|\mathcal{F}\rangle|}{\|H\|_1}}|\mathcal{F}\rangle|\Psi'(\mathcal{F},i)\rangle|a(\mathcal{F},i)\rangle|1\rangle.
\end{split}
\end{equation}
Using the fact that when $\langle\mathcal{F}'_i|H|\mathcal{F}\rangle\neq0$, we have $\Psi'(\mathcal{F},i)=\mathcal{F}'_i$ and $a(\mathcal{F},i)=0$ (see \eqref{bigpsidef} and \eqref{afuncdef}), we can rewrite the above expression as
\begin{equation}
\label{step3}
\begin{split}
    |\mathcal{F}\rangle\Bigg(&\sqrt{\frac{r}{\|H\|_1}}\sum_{|\mathcal{F}'\rangle\in Q(\mathcal{F})}\sqrt{\langle \mathcal{F}'|H|\mathcal{F}\rangle}|\mathcal{F}'\rangle|0\rangle|0\rangle\\
    &+\sum_{i=1}^k\sqrt{\frac{1}{k}-\frac{r|\langle \mathcal{F}'_i|H|\mathcal{F}\rangle|}{\|H\|_1}}|\Psi'(\mathcal{F},i)\rangle|a(\mathcal{F},i)\rangle|1\rangle\Bigg),
\end{split}
\end{equation}
Thus as noted in~\cite{berry12a}, the expression in parentheses in \eqref{step3} is equal to $|\phi_\mathcal{F}\rangle$ as given by \eqref{phidef_alt}, for $|\zeta_\mathcal{F}\rangle$ given by 
\begin{equation}
\label{zeta_psi_def}
\begin{split}
    |\zeta_\mathcal{F}\rangle=&\frac{1}{\sqrt{1-\frac{r\sigma_\mathcal{F}}{\|H\|_1}}}\\
    &\times\sum_{i=1}^k\sqrt{\frac{1}{k}-\frac{r|\langle \mathcal{F}'_i|H|\mathcal{F}\rangle|}{\|H\|_1}}|\Psi'(\mathcal{F},i)\rangle|a(\mathcal{F},i)\rangle.
\end{split}
\end{equation}
Hence, we have implemented $T$ using one query to $O_F$ and two queries to $O_H$, as in~\cite{berry12a}.
\end{proof}

In \cref{enumerator,matrix_element} we describe how to implement $O_F$ and $O_H$, respectively, in the compact mapping, for general interactions as described in \cref{interaction_def}.
These implementations, together with the construction in \cref{construction_lemma}, give us access to the optimal sparse Hamiltonian simulation technique afforded by qubitization~\cite{low19a}, as well as the other nearly-optimal techniques~\cite{berry15b,berry20a}.

In this paper, we focus on the application to simulating time-evolution, which is the goal of all of the sparse simulation papers we have cited~\cite{aharonov03a,childs03a,berry07a,childs10a,berry12a,berry14a,berry15a,berry15b,low17a,low19a,berry20a}.
A recent paper by some of the authors of the present work~\cite{kirby2020sparsevqe} demonstrated how to approximate ground state energies of sparse Hamiltonians using an extension of the variational quantum eigensolver (VQE), a hybrid quantum-classical algorithm that requires shorter circuits and is more noise-resilient than quantum algorithms for simulating time-evolution~\cite{peruzzo14a}.
Fock state oracles of the forms given in \eqref{enumerator_oracle_def}, \eqref{matrix_element_oracle_def} apply in this setting, which would allow us to implement VQE for second-quantized Hamiltonians in compact encoding.
The measurement scheme is given in~\cite{kirby2020sparsevqe}, and for an ansatz we could for example implement a version of Unitary Coupled Cluster~\cite{romero2018ucc} via oracle-based time evolutions generated by the Hamiltonian terms.

Because of the complexity of implementing the oracles, which we will present below in \cref{enumerator,matrix_element}, the circuit depths required for such an algorithm will be substantially longer than those used in VQE implementations appropriate for existing quantum computers.
Hence, they will require at least heavily error-mitigated or hardware-improved devices, and possibly fault-tolerance.
However, sparse VQE will still become possible before simulation of time-evolution, because it requires only a constant number of oracle queries (at most six~\cite{kirby2020sparsevqe}) per variational circuit, whereas simulation of time-evolution requires numbers of oracle queries that scale with the problem parameters, as in \eqref{query_complexity} or \eqref{query_complexity_time_dep} for example.

\subsection{Sparsity of general interactions in the Fock basis}

For the simulation methods described above to be efficient, we require the interactions to be sparse.
Consider an interaction specified as in \eqref{interaction_eq}, in $d$ dimensions with cutoffs as in \eqref{cutoff_def}.
The corresponding interaction Hamiltonian is given by \eqref{interaction_hamiltonian}.

To obtain an upper bound on the sparsity we may assume that the $f-g$ incoming particles in the interaction may be taken from any of the modes in the input state (recall that $f$ is the total number of external lines in the interaction, and $g$ is the number of outgoing lines).
In general, this upper bound is not tight, because it requires all modes in the input state to be the same up to momentum and to match all incoming lines in the interaction, but it will suffice to show that the sparsity is polynomial in the momentum cutoffs.
For $I$ the maximum number of modes in the input state (as in \cref{encoding}), this upper bound is
\begin{equation}
\label{incoming_assignments_bound}
    \Bigg(\hspace{-0.05in}\Bigg(\begin{matrix}I\\f-g\end{matrix}\Bigg)\hspace{-0.05in}\Bigg)
    =
    O\left(I^{f-g}\right),
\end{equation}
where the left-hand side denotes $I$ choose $f-g$ with replacement.
As discussed in \cref{hamiltonian}, the intrinsic quantum numbers of the outgoing particles are fixed by the interaction, but their momenta can take any values consistent with total momentum conservation.
For simplicity, we may upper bound this by counting all outgoing momentum assignments consistent with the cutoffs, of which there are
\begin{equation}
\label{outgoing_assignments_bound}
    O\left(\prod_{j=1}^d\big(\Ucutoff_j-\Lcutoff_j\big)^{g-1}\right),
\end{equation}
since we assign a momentum to each of the $g$ outgoing particles, but one of these is fixed by momentum conservation.

Hence, an upper bound on the total number of connected states to any given initial state under the interaction, and thus on the sparsity, is the product of \eqref{incoming_assignments_bound} and \eqref{outgoing_assignments_bound}:
\begin{equation}
\label{sparsity_general_cutoffs}
    k\le O\left(I^{f-g}\prod_{j=1}^d\big(\Ucutoff_j-\Lcutoff_j\big)^{g-1}\right)\le O\left(\Lambda_\text{max}^{d(f-1)}\right).
\end{equation}
The second upper bound is obtained by replacing $I$ with the total number of modes as a function of $\Lambda_\text{max}$, the maximum momentum cutoff.
This illustrates that even in this worst case, the sparsity is polynomial in the momentum cutoffs, for fixed dimension $d$ and number $f$ of external lines in the interaction.
The number of qubits $\qubitn$ is linear in $I$ up to logarithmic factors, so assuming $I$ is polynomial in the momentum cutoffs, $\qubitn$ is also polynomial in the momentum cutoffs, and thus the sparsity is polynomial in $\qubitn$.

\section{Enumerator oracle}
\label{enumerator}

In this section, we describe how to efficiently implement the oracle $O_F$, whose action is given by \eqref{enumerator_oracle_def}. 
In the first subsection, we describe some examples that illustrate all of the main techniques required for the general method.
In the second subsection, we describe the general method.
This description refers extensively to the details explained in the examples in the first subsection, so we strongly encourage the reader to begin with these.
Finally, in the third subsection we analyze the general method.

\subsection{Examples}
\label{enumerator_examples}

For an input state $|\mathcal{F}\rangle|i\rangle$, where $|\mathcal{F}\rangle$ is some (compact-encoded) Fock state and $i\in[k]$ for sparsity $k$, the enumerator oracle's action should be as follows: in an ancilla register, compute $|\mathcal{F}'_i\rangle$, the $i$th Fock state whose matrix element with $|\mathcal{F}\rangle$ is nonzero, and also uncompute $|i\rangle$.
Note that this action is assuming that $i$ does in fact index a nonzero matrix element (see \cref{sparsity}); the later examples will illustrate the possible situations where this may not be the case.

\begin{example}
\label{number_operator_example}
Consider a boson field on which our interaction is the number operator,
\begin{equation}
    H_I=\sum_\textbf{n}a^\dagger(\textbf{n})a(\textbf{n}).
\end{equation}
In the notation of \cref{interaction_def}, if we let `0' denote `boson', we would express this interaction as $\{(\text{`0'}),(\text{`0'})\}$ (meaning one boson in, one boson out), and the coefficient function is just $\beta(\textbf{n})=1$.
The interaction Hamiltonian $H_I$ maps any Fock state to itself, rescaled by a coefficient given by the total number of particles.
Hence, this interaction is diagonal and one-sparse (it maps each Fock state to at most one other Fock state).
Therefore, for any input $|\mathcal{F}\rangle|i=0\rangle$ (since $i$ takes only one value for a one-sparse interaction), the output of the enumerator oracle is just $|\mathcal{F}\rangle$, i.e.,
\begin{equation}
    O_F|\mathcal{F}\rangle|0\rangle=|\mathcal{F}\rangle|\mathcal{F}\rangle.
\end{equation}
Comparing this expression to \eqref{enumerator_oracle_def}, the reader will notice that we are suppressing the output $|a(\mathcal{F},i)\rangle$, but for the current example $a(\mathcal{F},i)$ is always $0$.
\end{example}

\begin{example}
\label{three_point_interaction_bosons_example}
Let us still consider a boson field, but now with the three-point interaction
\begin{equation}
    H_I=\sum_{\textbf{n}_1,\textbf{n}_2}a^\dagger(\textbf{n}_1+\textbf{n}_2)a(\textbf{n}_2)a(\textbf{n}_1)
\end{equation}
i.e., two incoming bosons annihilate to form a single outgoing boson whose momentum is the sum of the incoming momenta.
In the notation of \cref{interaction_def}, we would express this interaction as $\{(\text{`0'}),(\text{`0'},\text{`0'})\}$ (meaning two bosons in, one boson out), and the coefficient function is still just $\beta(\textbf{n}_1,\textbf{n}_2)=1$.
Even though this example appears only slightly more complicated than the number operator in \cref{number_operator_example}, it in fact will introduce almost all of the considerations we will require for completely general interactions.
We break the implementation of $O_F$ up into steps.

\textbf{\emph{Step 1 (identify incoming modes).}}
Given the input Fock state $|\mathcal{F}\rangle$, we need to use the input index $|i\rangle$ to determine the output Fock state $|\mathcal{F}'_i\rangle$.
The way we do this is to use $i$ to choose the two modes in $|\mathcal{F}\rangle$ that will contribute the incoming bosons in the interaction (they could come from the same mode).
For $I$ the maximum possible number of occupied modes (as in \cref{encoding}), we let $i$ index all pairs $J_i$ of mode indices: i.e.,
\begin{equation}
    \big\{J_i~|~i=1,2,...,I^2\big\}=\big\{\{j_1,j_2\}~|~j_1,j_2=1,2,...,I\big\}.
\end{equation}
So the first step in implementing $O_F$ is to compute $J_i$ in an ancilla register, as a function of $i$.

\textbf{\emph{Step 2 (remove incoming bosons).}}
Next, we find the momenta of the two modes indexed by $J_i$, and make sure that they have sufficient occupation to provide the incoming particles.
We can combine this with decrementing the occupation of those modes, as the first step towards constructing $|\mathcal{F}'_i\rangle$.
So, prior to beginning this step, we copy $|\mathcal{F}\rangle$ to a second qubit register that will become the output register containing $|\mathcal{F}'_i\rangle$ at the end of the implementation.
Note that this copying of $|\mathcal{F}\rangle$ is allowed because it is just copying in the (compact-encoded) Fock basis, which can be implemented via qubitwise CNOTs from the input $|\mathcal{F}\rangle$ to the new copy register assuming this is initially in the all-zeroes state.

On the copy of $|\mathcal{F}\rangle$, we first find the $j_1$th mode (where $J_i=(j_1,j_2)$), and check whether it has nonzero occupation.
If it does not, then we need some way to record the fact that $i$ does not index a valid nonzero matrix element.
The way we do this is to maintain an ancilla register called the ``flag" register, whose value is initially zero and should remain zero at the end of the implementation of $O_F$ if and only if $i$ indexes a valid nonzero matrix element.
So, to check whether the $j_1$th mode has nonzero occupation, we add one to the flag register controlled on $j_1$th mode having zero occupation.
Now we may proceed as if the mode does have nonzero occupation, knowing that if the flag register is nonzero at the end of the implementation we will simply reverse the whole procedure.

Next we decrement the occupation of the $j_1$th mode by one, and record its momentum in an ancilla register that encodes $\textbf{n}_1$, the momentum of the first incoming boson.
If the occupation of the mode is now zero, we should also set the remaining qubits encoding the mode to all zeroes (which we can do by applying CNOTs controlled on the corresponding qubits in the original version of $|\mathcal{F}\rangle$~--- recall that we are currently operating on the copy).
Also controlled on the occupation of the current mode being zero, we store its index in the first entry of an ancilla register $E$, which we will use later.

We now repeat this whole procedure (including checking the occupation) for the $j_2$th mode, recording its momentum in another ancilla register that encodes $\textbf{n}_2$.
If this mode is left empty, we store its index in the second entry of the ancilla register $E$.
Note that if both modes are the same, i.e., $j_1=j_2$, then at this second step we will add one to the flag register if the initial occupation of that mode was not at least two, since one boson has already been removed from the mode.

\textbf{\emph{Step 3 (reorder modes).}}
Once we are done with steps 1 and 2, we have $\textbf{n}_1$ and $\textbf{n}_2$ recorded in ancilla registers, and we have decremented the corresponding mode occupations in the copy of $|\mathcal{F}\rangle$.
However, it is possible that up to two modes in the copy of $|\mathcal{F}\rangle$ may have been left empty after their occupations were decremented.
Since the compact encoding only stores occupied modes, and encodes them in the first $J$ mode registers $X_i$ (see \eqref{qubit_state_encoding}), we need to move any empty modes to the end of the encoding.

To do this, we use the ancilla register $E$ that we introduced above.
This contains two entries $E_1$ and $E_2$; $E_1=j_1$ if the $j_1$th mode was left empty after removing the incoming particles to the interaction, and otherwise $E_1$ remains in its initial state (which should be chosen to be different from any of the values encoding indices).
$E_2$ was similarly set according to the $j_2$th mode.

Using these, we reorder the modes as follows.
Iterate over $j=1,2,...,I-1$ (the mode indices).
For each $j$, swap the $j$th register and the $(j+1)$th register controlled on $E_1\le j$ (and on $E_1$ actually encoding a valid mode index).
If the $j_1$th mode is emptied, the first time the above control condition will be satisfied is when $j=E_1=j_1$, so this will swap the emptied $j_1$th mode register with the $(j_1+1)$th mode register.
Next we move to $j=j_1+1\ge E_1$, so the $(j_1+1)$th mode register (which now contains the empty mode) gets swapped with the $(j_1+2)$th mode register, and so forth until the empty mode has been moved to the last mode register.

We now repeat this procedure for $E_2$ in order to move the $j_2$th mode to the end of the mode registers if it was emptied.
The only caveat with this step is that, if both the $j_1$th and $j_2$th modes were emptied and $j_1<j_2$, then after the first swapping sequence (for $E_1=j_1$), the empty mode that was initially in position $j_2$ is now in position $j_2-1$, since the $j_1$th mode was swapped out from before it.
Therefore, before repeating the procedure for $E_2$, we should subtract one from $E_2$ controlled on $E_1<j_2$.

\textbf{\emph{Step 4 (insert outgoing boson).}}
After we have completed all of the above steps, the copy of $|\mathcal{F}\rangle$ has had the two incoming bosons to the interaction removed, and the modes have been reordered if necessary so that the $J$ occupied modes are encoded in the first $J$ mode registers.
The momenta $\textbf{n}_1,\textbf{n}_2$ of the incoming bosons are also recorded in ancilla registers.
All that remains is to insert the new boson with momentum $\textbf{n}_1+\textbf{n}_2$.

To do this, we first iterate over the already occupied modes, checking whether each one has momentum $\textbf{n}_1+\textbf{n}_2$ and incrementing its occupation if so.
We also flip a single ancilla qubit from $|0\rangle$ to $|1\rangle$ if we find such a mode, to record the fact that we have inserted the new boson.
If at the end of the iterations, this ancilla qubit is still $|0\rangle$, then the new boson needs to be inserted as a new mode (so the following operations should be controlled on this).
In this case, we first find the location where the new mode should be inserted: to do this, we iterate over $j=1,2,...,I$, checking whether $\textbf{n}_1+\textbf{n}_2$ is greater than the momentum of the $(j-1)$th mode and less than the momentum of the $j$th mode.
This will only be true for a single value of $j$, which we can call $j'$, so we record $j'$ in an ancilla register.

Then we iterate over the mode indices in reverse order, i.e., $j=I-1,...,1$, for each $j$ swapping the $j$th and $(j+1)$th mode registers controlled on $j'\le j$.
Since the maximum possible number of occupied modes is $I$ and we are about to insert a new mode, prior to this sequence of swaps the $I$th mode register is guaranteed to be unencoded.
Hence, when $j=I-1$, we swap the $(I-1)$th and $I$th mode registers, moving the unencoded register to location $I-1$.
We then proceed to $j=I-2$, swapping the $(I-2)$th and $(I-1)$th mode registers and thus moving the unencoded register to location $I-2$, and so forth.
The last swap we perform is when $j=j'$, so once we are done with the full sequence of swaps the unencoded register is in location $j'$.
Therefore, we can simply set its momentum to $\textbf{n}_1+\textbf{n}_2$ and set its occupation to one, and we are done.

\textbf{\emph{Step 5 (uncompute ancillas).}}
Once all of the above operations are complete, the copy of $|\mathcal{F}\rangle$ has been transformed into the desired output Fock state $|\mathcal{F}'_i\rangle$, so all that remains is to uncompute the ancillas.
This could be done simply by copying $|\mathcal{F}'_i\rangle$ (in the Fock basis) and then exactly reversing all of the prior operations, but we also want to uncompute $|i\rangle$, the input index.
In order to accomplish this, we uncompute $|i\rangle$ using the register encoding $J_i$.
But now we cannot uncompute $J_i$ using $i$, since this value has been uncomputed, so we instead need to uncompute $J_i$ using the registers it was used to compute.
The details of this are tedious, but since $|\mathcal{F}\rangle$ and $|\mathcal{F}'_i\rangle$ together contain enough information to determine the values of all of the ancillas that were used to compute $|\mathcal{F}'_i\rangle$, we can use those values to uncompute the ancillas via similar operations to those used to compute $|\mathcal{F}'_i\rangle$.

Also, if we reach the end of the procedure and the flag register is nonzero, then we know that the index $i$ did not in fact correspond to a valid matrix element.
In this case, the desired output as given in \eqref{enumerator_oracle_def} is $|\mathcal{F}\rangle|\mathcal{F}\rangle|i\rangle$.
Hence, we want to keep $|i\rangle$ (which is $|a(\mathcal{F},i)\rangle$ in this case), completely reverse the rest of the computation, and then just copy $|\mathcal{F}\rangle$ itself to the output Fock state register.

This completes the implementation of $O_F$ for this example.

\end{example}

\begin{example}
\label{four_point_interaction_bosons_example}
Let us again consider a boson field, but now with the four-point interaction
\begin{equation}
    H_I=\sum_{\textbf{n}_1,\textbf{n}_2,\textbf{n}_3,\textbf{n}_4}a^\dagger(\textbf{n}_4)a^\dagger(\textbf{n}_3)a(\textbf{n}_2)a(\textbf{n}_1),
\end{equation}
where the sum runs over all momentum conserving combinations, i.e., $\textbf{n}_1+\textbf{n}_2=\textbf{n}_3+\textbf{n}_4$.
Many of the steps to implement the $O_F$ oracle for this interaction are the same as in \cref{three_point_interaction_bosons_example}, so instead of going through the entire procedure again, we will just describe what needs to change.

Steps 1 through 3, in which we identify the incoming modes, decrement their occupations, and reorder the modes if some of them are left empty, are the same as in \cref{three_point_interaction_bosons_example}.
However, the incoming momenta no longer uniquely determine the outgoing momenta, since for a given value of $\textbf{n}_1+\textbf{n}_2$ there are multiple values of $\textbf{n}_3$ and $\textbf{n}_4$ that satisfy momentum conservation.
Therefore, our input index $i$ needs to do more work than just to specify $J_i$ (the incoming mode indices).
In particular, after the total incoming momentum 
\begin{equation}
    \textbf{Q}\equiv\textbf{n}_1+\textbf{n}_2
\end{equation}
has been determined, we need to use $i$ to determine how this momentum should be split up between the outgoing bosons.

To do this, we use a classically-precomputed lookup table $A(\textbf{Q},i)$ that maps any possible value of the total momentum $\textbf{Q}$ together with an index value $i$ to a partition of the momentum among the outgoing particles.
For example, if we are in a 1+1D light-front field theory (see \cref{lf_application}) with momentum cutoffs $\Lcutoff=1$, $\Ucutoff=K=5$, then we could take $A(\textbf{Q},i)$ to be
\begin{equation}
\label{A_example}
\begin{split}
    A=
    \{&\\
    &(2,0)\mapsto\{1,1\},\\
    &(3,0)\mapsto\{2,1\},\\
    &(4,0)\mapsto\{3,1\},\\
    &(4,1)\mapsto\{2,2\},\\
    &(5,0)\mapsto\{4,1\},\\
    &(5,1)\mapsto\{3,2\}\\
    \}.
\end{split}
\end{equation}
In other words, for each fixed value of $\textbf{Q}$, $i$ indexes the possible partitions of $\textbf{Q}$ into two parts satisfying the momentum cutoffs, and $A(\textbf{Q},i)$ returns these values.
So, for example, if $\textbf{Q}=4$ and $i=1$, then 
\begin{equation}
    A(\textbf{Q},i)=A(4,1)=\{2,2\}.
\end{equation}

So, given the actual value of $\textbf{Q}$, which is stored in some ancilla register, and the value of $i$, which is one of the quantum inputs, we need to compute $A(\textbf{Q},i)$ in an ancilla register.
To do this, we classically iterate over the possible values $(\textbf{Q}',i')$, for each one setting the ancilla register to $A(\textbf{Q}',i')$ controlled on ${(\textbf{Q},i)=(\textbf{Q}',i')}$.
When this iteration is complete, we will have the outgoing momenta stored in this ancilla register.

Since we are now using $i$ both to specify $J_i$ and $A(\textbf{Q},i)$, we need to keep these independent.
To do this, if $A(\textbf{Q},i)$ requires at most $a$ distinct values of $i$, then we can let $A(\textbf{Q},i)$ be a function of ${i\text{ mod }a}$ and $J_i$ be a function of $\lfloor i/a\rfloor$.
For example, the instance of $A(\textbf{Q},i)$ given in \eqref{A_example} only requires two distinct values of $i$, so for this case we could let
\begin{equation}
    A(\textbf{Q},i)\mapsto A(\textbf{Q},i\text{ mod }2)
\end{equation}
and
\begin{equation}
    J_i\mapsto J_{\lfloor i/2\rfloor}.
\end{equation}

The only additional consideration is that, as illustrated in \eqref{A_example}, depending on the value of $\textbf{Q}$ not all of the values of $i\text{ mod }a$ may be used to specify outgoing momentum assignments via $A(\textbf{Q},i)$.
If $i\text{ mod }a$ takes one of these unused values, then this is just another instance of $i$ not indexing a valid matrix element, so we should add one to the flag register.
This would happen, for example, if we obtained the inputs $\textbf{Q}=2$, $i\text{ mod }2=1$ for $A(\textbf{Q},i\text{ mod }2)$ as given in \eqref{A_example}.

Once we have specified the two outgoing momenta $\textbf{n}_3$ and $\textbf{n}_4$, inserting them in the outgoing Fock state just requires applying step 4 in \cref{three_point_interaction_bosons_example} twice.
Step 5 is then also the same as for \cref{three_point_interaction_bosons_example}, and that completes the implementation of $O_F$ for this example.

\end{example}

\begin{example}
\label{bosons_fermions_four_point_example}
Let us now, finally, consider an interaction including fermions as well as bosons:
\begin{equation}
    H_I=\sum_{\textbf{n}_1,\textbf{n}_2,\textbf{n}_3,\textbf{n}_4}a_1^\dagger(\textbf{n}_4)a_0^\dagger(\textbf{n}_3)a_1(\textbf{n}_2)a_0(\textbf{n}_1),
\end{equation}
where subscript $0$ indicates boson and subscript $1$ indicates fermion, and the sum runs over all momentum conserving combinations, i.e., $\textbf{n}_1+\textbf{n}_2=\textbf{n}_3+\textbf{n}_4$.
Hence this interaction is an incoming boson and fermion, and an outgoing boson and fermion.

Many of the elements of the implementation of $O_F$ are the same as in the previous two examples.
One change is that the orders of the values of $J_i$ and $A(\textbf{Q},i)$ now matter, since the two incoming modes are now distinguishable, as are the two outgoing modes.
Also, we must now check that for $J_i=(j_1,j_2)$, the $j_1$th mode is bosonic and the $j_2$th mode is fermionic, adding one to the flag register if either is not.
The remainder of identifying and removing the incoming particles, reordering the modes, and computing the outgoing momenta are the same as in the prior examples.

When inserting the outgoing fermion we must also alter the procedure.
When we iterate over the modes to check whether any match the new fermion to be inserted, instead of incrementing its occupancy if we find a mode that matches (as we would for a boson), we add one to the flag register, because fermionic modes cannot have occupancy greater than one.
We then proceed with inserting the fermion as a new mode in exactly the same way as for bosons, and complete the rest of the procedure exactly as in \cref{three_point_interaction_bosons_example}.
Note that when computing the matrix element oracle $O_H$, we will have to additionally treat fermions and bosons differently because of their different commutation relations, but for the enumerator oracle this is not relevant.

\end{example}

\subsection{General method}

\begin{figure*}[htp]
\centering
\includegraphics[width=6.5in]{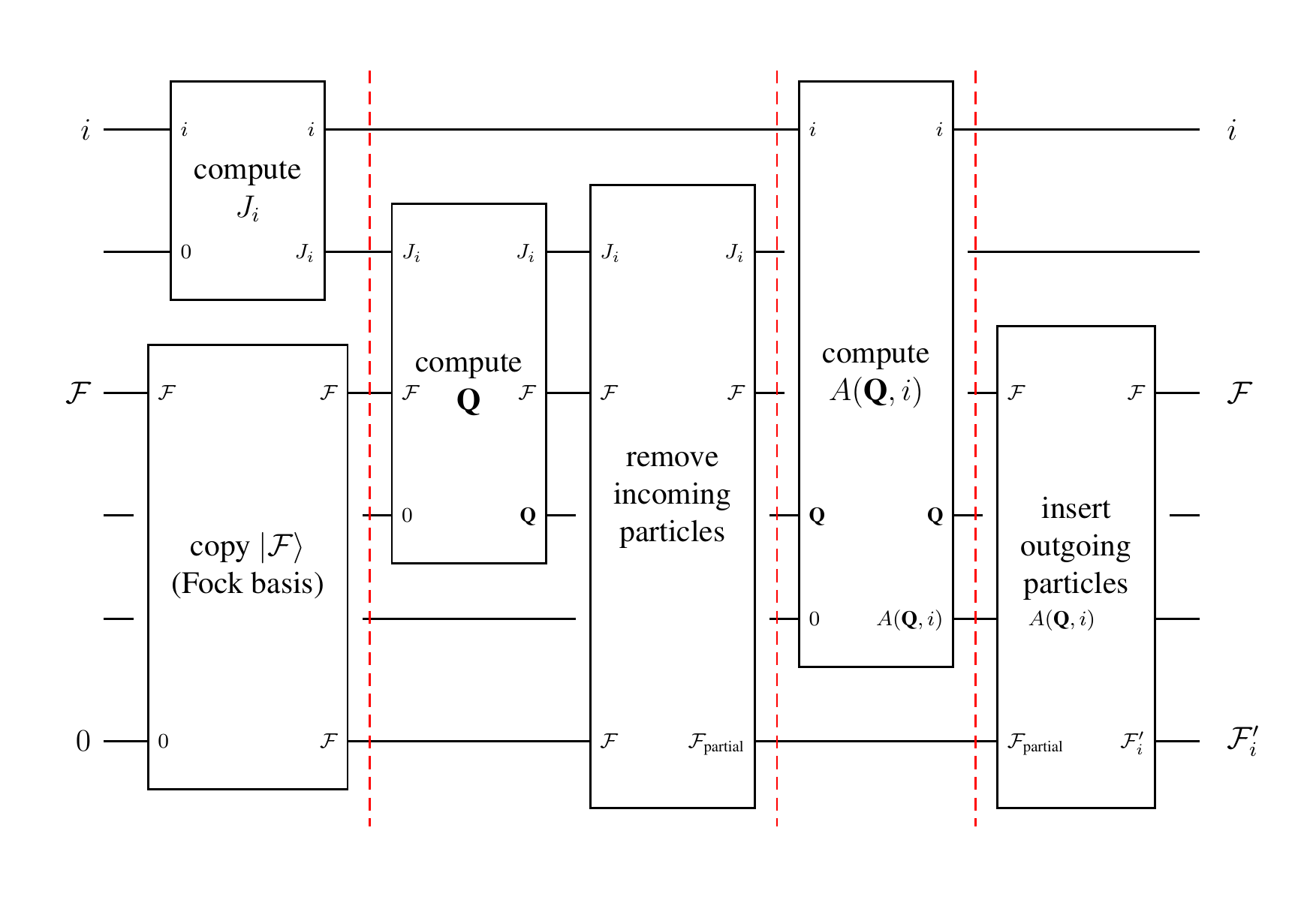}
\vspace{-0.5in}
\caption{Schematic of the circuit to implement the enumerator oracle. Labeled wires are input and output registers, and unlabeled wires are ancilla registers (each initially in the all $0$s state). The inputs and outputs are $|\mathcal{F}\rangle$ (the incoming Fock state), $|i\rangle$ (the sparsity index), and $|\mathcal{F}'_i\rangle$ (the outgoing Fock state), all as in \eqref{enumerator_oracle_def}. The intermediate quantities computed are $J_i$ (the list of modes in $|\mathcal{F}\rangle$ from which the incoming particles are taken, as in \eqref{J_i_values}), $\textbf{Q}$ (the total incoming and outgoing momentum), and $A(\textbf{Q},i)$ (the list of momenta of the outgoing particles, as in \eqref{A_form}). The circuit hides some additional ancillas, does not show uncomputation of the ancillas, and shows the action when $i$ is a valid index for an outgoing state. At the points marked by red dashed lines we have to check whether this is true, as described in detail in the text.
}
\label{fig:circuit}
\end{figure*}

We assume that the input is in the form $|\mathcal{F}\rangle|i\rangle$, where $|\mathcal{F}\rangle$ is some compact-encoded Fock state, and $i\in[k]$ where $k$ is the sparsity.
The $f$-point interaction is specified as in \cref{interaction_def}, i.e., as a set of $g$ outgoing lines identified by $(q_1,q_2,...,q_g)$, and a set of $f-g$ incoming lines identified by $(q_{g+1},q_{g+2},...,q_f)$.
The momentum cutoffs are as in \eqref{cutoff_def}: for each momentum $\textbf{n}$, each component $n_j$ must satisfy
\begin{equation}
    \Lcutoff_j\le n_j\le\Ucutoff_j,
\end{equation}
where $j$ runs over the dimensions.
All of the main ideas for the implementation of $O_F$ were introduced by the examples in \cref{enumerator_examples}, mostly in \cref{three_point_interaction_bosons_example}, so we simply indicate how to appropriately generalize these ideas in order to describe the method for arbitrary interactions.

\textbf{\emph{Step 1 (identify incoming modes).}}
This step is the same as in \cref{three_point_interaction_bosons_example,four_point_interaction_bosons_example}, except that the possible sets of incoming modes now have size $f-g$:
\begin{equation}
\label{J_i_values}
    J_i\in\{1,2,...,I\}^{f-g}.
\end{equation}
We therefore have to check that for $J_i=(j_1,j_2,...,j_{f-g})$, the $j_k$th mode in $|\mathcal{F}\rangle$ matches the identifying information $q_{g+k}$ of the $k$th incoming line in the interaction.
This generalizes the step in \cref{bosons_fermions_four_point_example} where we check that the $j_1$th mode is bosonic and the $j_2$th mode is fermionic.

\textbf{\emph{Step 2 (remove incoming particles and reorder modes).}}
This step is implemented exactly as in the examples, just with more repetitions of the removal procedure as we decrement the occupation in modes $j_1$, $j_2$, through $j_{f-g}$.
We store the momenta $\textbf{n}_{g+1},\textbf{n}_{g+2},...,\textbf{n}_f$ of the incoming lines in ancilla registers, and compute their sum $\textbf{Q}$.
Reordering modes is also exactly as in the examples; the list $E$ of emptied modes must now contain $f-g$ entries.

\textbf{\emph{Step 3 (compute outgoing momenta).}}
The classical lookup table $A(\textbf{Q},i)$ maps values of $\textbf{Q}$ and $i$ to ordered sets of $g$ outgoing momenta:
\begin{equation}
\label{A_form}
    A(\textbf{Q},i)=(\textbf{n}_1,\textbf{n}_2,...,\textbf{n}_g)
\end{equation}
such that
\begin{equation}
\label{momentum_conservation_general}
    \sum_{k=1}^g\textbf{n}_k=\textbf{Q}.
\end{equation}
As in \cref{four_point_interaction_bosons_example}, we classically iterate over the possible values $\textbf{Q}'$ and $i'$ of $\textbf{Q}$ and $i$, and implement a quantum operation that encodes $A(\textbf{Q}',i')$ in an ancilla register controlled on $(\textbf{Q}',i')=(\textbf{Q},i)$.

The only difference is that now $\textbf{Q}$ and $i$ can take more values.
$\textbf{Q}$ can be any momentum that is the sum of $f-g$ momenta consistent with the cutoffs $\Lcutoff_j,\Ucutoff_j$ in each dimension $j$.
For a given $\textbf{Q}$, outgoing momenta can be any set satisfying \eqref{momentum_conservation_general}, so $i$ must provide enough distinct values to distinguish these assignments for whichever value of $\textbf{Q}$ gives the most of them.
We will provide a detailed analysis of this later.

\textbf{\emph{Step 4 (insert outgoing particles).}}
This step is the same as in the examples in \cref{enumerator_examples}, except that we must now repeat the insertion procedure $g$ times, once for each of the $g$ outgoing modes.
The momenta of the outgoing modes are given by $\textbf{n}_1,\textbf{n}_2,...,\textbf{n}_g$, which we computed in step 3, and their identifying information (particle types and quantum numbers) is given by $q_1,q_2,...,q_g$ in the interaction specification.

\textbf{\emph{Step 5 (uncompute ancillas).}}
This step is the same as in the examples in \cref{enumerator_examples}.

This completes the implementation of $O_F$ for a general interaction.
A schematic for the circuit is shown in \cref{fig:circuit}.

\subsection{Analysis}

We will analyze the above algorithm in terms of the number of log-local operations required.
The specific log-local operations of interest are actions on constant numbers of mode registers $X_j$, either in encoded states or in ancilla registers.
These are log-local because each mode register contains logarithmically-many qubits in the momentum and occupation number cutoffs (see \eqref{moden}). The log-local operations we used are all controlled arithmetic operations. The problem of compiling such operations into primitive gates can be addressed independently, and is well-studied (see for example~\cite{JavadiAbhari2014scaffcc}). The choice of primitive gate set to compile into is also hardware-specific. Hence, we express our gate counts in terms of the log-local operations.

We analyze each of the steps outlined in the previous section.
Step 1 requires controlling on the possible values of $J_i$, leading to a number of log-local operations that scales with the number of possible values of $J_i$.
By \eqref{J_i_values}, the number of possible values of $J_i$ is upper bounded by $I^{f-g}$ (recall that $I$ is the maximum possible number of occupied modes), so
\begin{equation}
\label{step_1_cost}
    O\left(I^{f-g}\right)
\end{equation}
is an upper bound on the number of log-local operations required to implement step 1.
Recall that $f$ and $g$ are constant, so \eqref{step_1_cost} is polynomial in $I$.

Step 2 requires finding the modes whose indices match indices in $J_i$, decrementing their occupations, and reordering the modes: these are implemented via a constant number of simultaneous iterations over the $f-g$ entries in $J_i$, and over the $I$ modes in the copy of $|\mathcal{F}\rangle$.
Thus
\begin{equation}
    O\left(I(f-g)\right)
\end{equation}
is an upper bound on the number of log-local operations required to implement step 2.

Step 3 requires controlling on the pairs of possible values of $\textbf{Q}$ and $i$ that give distinct values of $A(\textbf{Q},i)$.
The number of such pairs is the same as the number of possible distinct values of $A(\textbf{Q},i)$.
These values have the form \eqref{A_form}, so the number of possible values is upper bounded by the number of possible values for each entry, raised to power $g$.
Each entry is the momentum of a single particle, so if we take $\Lambda_\text{max}$ to be the maximum momentum cutoff (in magnitude) over all dimensions, the number of possible values for each entry in $A(\textbf{Q},i)$ is $O(\Lambda_\text{max}^d)$ (recall that $d$ is the spatial dimension).
Hence,
\begin{equation}
    O\left(\Lambda_\text{max}^{dg}\right)
\end{equation}
is an upper bound on the number of distinct values of $A(\textbf{Q},i)$, and thus also an upper bound on the number of log-local operations required to implement step 3.

Step 4 requires a constant number of simultaneous iterations over the $g$ outgoing modes (determined by the value of $A(\textbf{Q},i)$ and $q_1,q_2,...,q_g$ as specified by the interaction), and over the $I$ modes in the copy of $|\mathcal{F}\rangle$.
Thus
\begin{equation}
    O\left(Ig\right)
\end{equation}
is an upper bound on the number of log-local operations required to implement step 4.

Step 5, uncomputing the ancillas, at worst doubles the cost of the full algorithm, so we may ignore it in the scaling.
The costs of steps 2 and 4 are subsumed by the costs of steps 1 and 3, so the total number of log-local operations required to implement the enumerator oracle and compute the inputs to the matrix element function is
\begin{equation}
\label{cost_scaling}
    O\left(I^{f-g}+\Lambda_\text{max}^{dg}\right).
\end{equation}
Hence, the number of log-local operations required to implement the enumerator oracle and compute the inputs to the matrix element function is polynomial in the momentum cutoff $\Lambda_\text{max}$, the number $I$ of mode registers, and the number of qubits (since this is linear in $I$ and logarithmic in the other parameters~--- see \eqref{qubitn}).

\section{Matrix element oracle}
\label{matrix_element}

The oracle $O_H$ defined in \eqref{matrix_element_oracle_def} calculates a matrix element of the interaction Hamiltonian $H_I$ (given by \eqref{interaction_hamiltonian}) to some desired precision.
The quantum input is a pair of compact-encoded Fock states $|\mathcal{F}\rangle$ and $|\mathcal{F}'\rangle$, taken to be the incoming and outgoing states in the interaction, respectively.
As described in the proof of \cref{construction_lemma}, above, $O_H$ is only implemented when $|\mathcal{F}'\rangle=|\mathcal{F}'_i\rangle$ for some $i$ (recall that $|\mathcal{F}'_i\rangle$ is the $i$th connected state to $|\mathcal{F}\rangle$), so we may assume that the matrix element of $|\mathcal{F}\rangle$ and $|\mathcal{F}'\rangle$ is nonzero (or that it is zero and that fact is recorded by the register $a(\mathcal{F},i)$ being nonzero~--- see \eqref{afuncdef}).

The value of the matrix element is given by its coefficient $\beta(\{\textbf{n}_i\})$ as in \eqref{interaction_hamiltonian} multiplied by any factors coming from the ladder operators.
As usual, applying a creation operator to a mode containing $w$ particles contributes a factor of $\sqrt{w+1}$, while applying an annihilation operator contributes a factor of $\sqrt{w}$.
In order to enforce antisymmetrization of fermions and antifermions, each (anti)fermionic ladder operator also contributes a factor of $\pm1$ determined by the parity of the number of particles of the same type encoded in mode registers preceding the mode register acted upon by the ladder operator (in the canonical ordering established in \cref{encoding}).

Consider a general interaction, with incoming lines $q_{g+1},q_{g+2},...,q_f$ and outgoing lines $q_1,q_2,...,q_g$.
This interaction connects Fock states $|\mathcal{F}\rangle$ and $|\mathcal{F}'\rangle$ when there is some assignment of momenta $\{\textbf{n}_i\}$ to the incoming and outgoing lines that conserves momentum, i.e., $\sum_{i=1}^g\textbf{n}_i=\sum_{i=g+1}^f\textbf{n}_i$, such that
\begin{equation}
    \langle\mathcal{F}'|\left(\prod_{i=1}^ga^\dagger_{q_i}(\textbf{n}_i)\right)\left(\prod_{i=g+1}^fa_{q_i}(\textbf{n}_i)\right)|\mathcal{F}\rangle\neq0.
\end{equation}
This is the case if and only if $\{(q_i,\textbf{n}_i)~|~i=1,2,...,g\}$ are the extra particles in $|\mathcal{F}'\rangle$ (and not in $|\mathcal{F}\rangle$), and $\{(q_i,\textbf{n}_i)~|~i=g+1,g+2,...,f\}$ are the extra particles in $|\mathcal{F}\rangle$ (and not in $|\mathcal{F}'\rangle$).
When this condition holds,
\begin{equation}
\begin{split}
    &\langle\mathcal{F}'|\left(\prod_{i=1}^ga^\dagger_{q_i}(\textbf{n}_i)\right)\left(\prod_{i=g+1}^fa_{q_i}(\textbf{n}_i)\right)|\mathcal{F}\rangle\\
    &=\pm\sqrt{\left(\prod_{i=1}^gw'_i\right)\left(\prod_{i=g+1}^fw_i\right)},
\label{ladder_ops_value}
\end{split}
\end{equation}
where the $\pm$ is set by fermion/antifermion antisymmetrization, each $w_i$ (for $i=g+1,g+2,...,f$) is the occupation of the mode $(q_i,\textbf{n}_i)$ in
\begin{equation}
    \left(\prod_{j=i+1}^fa_{q_j}(\textbf{n}_j)\right)|\mathcal{F}\rangle,
\end{equation}
and each $w'_i$ (for $i=1,2,...,g$) is the occupation of the mode $(q_i,\textbf{n}_i)$ in
\begin{equation}
    \left(\prod_{j=i}^ga^\dagger_{q_j}(\textbf{n}_j)\right)\left(\prod_{j=g+1}^fa_{q_j}(\textbf{n}_j)\right)|\mathcal{F}\rangle.
\end{equation}
In other words, if multiple creation or annihilation operators act on the same mode, for each the corresponding $w_i$ or $w'_i$ should be the occupation of the mode immediately before the annihilation or after the creation.

\begin{example}
Consider $a_0(2)^\dagger a_0(1)^2$, i.e., annihilation of two identical bosons with momentum one followed by creation of a boson with momentum two.
If the input state is
\begin{equation}
    |\mathcal{F}\rangle=|(0,1,5)\rangle,
\end{equation}
i.e., five bosons of momentum one and nothing else, then the output state is
\begin{equation}
    |\mathcal{F}'\rangle=|(0,1,3),(0,2,1)\rangle,
\end{equation}
i.e., three bosons of momentum one and one boson of momentum two.
Hence $w'_1=1$ (since the created momentum-two boson is in its own mode), $w_2=4$, and $w_3=5$ (since the momentum-one mode has occupation 5 when the first boson is annihilated and occupation 4 when the second boson is annihilated).
Thus for this example, \eqref{ladder_ops_value} becomes
\begin{equation}
\label{ladder_ops_value_ex}
    \langle\mathcal{F}'|a(2)^\dagger a(1)^2|\mathcal{F}\rangle=\sqrt{w'_1w_2w_3}=\sqrt{20}.
\end{equation}
\end{example}
Similarly, the value of the parity factor $\pm1$ in \eqref{ladder_ops_value} is the product of the parity factors due to the ladder operators at the times when they are applied.
The coefficient $\beta$ is a function of the $\textbf{n}_i$, so the complete value of the matrix element is 
\begin{equation}
\label{matrix_element_expression}
    \pm\beta(\{\textbf{n}_i\})\sqrt{\left(\prod_{i=1}^gw'_i\right)\left(\prod_{i=g+1}^fw_i\right)}.
\end{equation}

Recall that this is all assuming that $|\mathcal{F}'\rangle$ is connected to $|\mathcal{F}\rangle$ by the interaction, and that $\{\textbf{n}_i\}$ is the corresponding assignment of momenta to the external lines in the interaction.
But as pointed out in the first paragraph of this section, we may assume that we only have to evaluate the matrix element for pairs of states that are the output of the enumerator oracle, and hence are connected.
Therefore, computing the matrix element of $|\mathcal{F}\rangle,|\mathcal{F}'\rangle$ requires two steps:
\begin{enumerate}
    \item Find the momenta of the extra particles in each state, the occupations of the corresponding modes (accounting for the case when multiple bosons in the same mode are created or annihilated), and the parities of the preceding modes for particles of the same type (for fermions and antifermions).
    \item Evaluate \eqref{matrix_element_expression}.
\end{enumerate}

When we apply the enumerator oracle to determine $|\mathcal{F}'_i\rangle$ given $|\mathcal{F}\rangle$ and $i\rangle$, we can obtain the first step above along the way.
In particular, the set of indices $J_i$ \eqref{J_i_values} identifies the set of extra particles in $|\mathcal{F}\rangle$, and $A(\textbf{Q},i)$ is the set of momenta of the extra particles in $|\mathcal{F}'_i\rangle$.
The occupations of the corresponding modes in $|\mathcal{F}'_i\rangle$ are identified when the new particles are inserted to construct $|\mathcal{F}'_i\rangle$.
The parities for fermions and antifermions can be obtained by simply counting the numbers of preceding modes with the same particle type (the particle types are defined by the interaction), since the positions of the modes that the ladder operators act on are specified explicitly by $J_i$ (for the incoming particles), and in the course of inserting the outgoing particles.
Therefore, by the time we have obtained $|\mathcal{F}'_i\rangle$ in the course of implementing $O_F$, we can also complete step 1 of implementing $O_H$ above.
Thus we can execute $O_H$ as many times as desired by implementing step 2 above, as long as we do so prior to uncomputing the ancillas used to compute $O_F$.

In order to implement the quantum walk operator $T$ (see \eqref{newTdef}), we require two applications of $O_H$, one to compute the matrix element, and another to uncompute the matrix element after performing a rotation controlled on it (see step 3 in the proof of \cref{construction_lemma}, above).
There is no problem in putting off uncomputing the ancillas used in the computation of $O_F$ until after the controlled rotation has been executed.
Thus, we can perform both applications of $O_H$ simply by executing step 2 above, with the inputs given by these ancillas.
In other words, we can include all necessary applications of $O_H$ in our implementation of $O_F$, without needing to recompute the extra particles in each state $|\mathcal{F}\rangle,|\mathcal{F}'\rangle$.

The implementation of step 2 above, i.e., the actual evaluation of the matrix element as in \eqref{matrix_element_expression}, depends on the specific functional form of $\beta(\{\textbf{n}_i\})$.
However, we can make some general statements.
The matrix element expression \eqref{matrix_element_expression} is a function of $2f$ variables, $\{\textbf{n}_i\}$ and $w'_i,w_i$.
Each of the $\textbf{n}_i$ is a $d$-dimensional vector whose entries are constrained by the cutoffs \eqref{cutoff_def}, so if $\Lambda_\text{max}$ is the maximum magnitude of any cutoff, $\textbf{n}_i$ takes $O(\Lambda_\text{max}^d)$ values and is encoded in $O(d\log\Lambda_\text{max})$ qubits for each $i$.
Each of the $w_i$ and $w'_i$ is a positive integer upper bounded by $W$, where $W$ is the occupation number cutoff, so each can be encoded in $O(\log W)$ qubits.
Elementary arithmetic operations can be implemented as sequences of NOT, CNOT, and Toffoli gates with depth polynomial in the number of qubits of the inputs~\cite{vedral96a,JavadiAbhari2014scaffcc}.

Thus, assuming that the matrix element can be expressed as a fixed combination of elementary arithmetic operations, evaluating it requires
\begin{equation}
\label{matrix_element_cost_scaling}
    O\big(d\,f\,\text{polylog}(\Lambda_\text{max})+f\,\text{polylog}(W)\big)
\end{equation}
NOT, CNOT, and Toffoli gates.
In other words, for fixed interactions in fixed dimension, the entire $O_H$ can be executed by using the ancilla values computed during implementation of $O_F$, with gate count overhead that is polylogarithmic in the momentum and occupation number cutoffs.

\section{Analysis and Applications}
\label{applications}

In this section, we explain how to simulate several example models in 1+1D using the tools we have described above.
For each model, we also provide a comparison of the number of gates required for equal-time versus light-front formulations of relativistic quantum field theory.

The gates we are counting are log-local operations, i.e., operations on constant numbers of registers encoding single modes.
As noted above, we do not compile these operations all the way into primitive gates, because optimizing such compilation is an independent problem and is itself the subject of extensive study (see for example~\cite{JavadiAbhari2014scaffcc}), as well as being hardware-specific.
The gate counts we provide are also only for implementing the enumerator oracle and obtaining the inputs to the coefficient function, since as explained in \cref{matrix_element}, once these steps are complete computing the value of the matrix element requires a number of additional gates that is polylogarithmic in the momentum and occupation cutoffs (see \eqref{matrix_element_cost_scaling}).

\subsection{Free boson and fermion theory}\label{freeoracle}

We begin by examining free theories for bosons or fermions before moving on to interacting theories. The Hamiltonians we consider are linear combinations of number operators, and thus diagonal and $1$-sparse. More general free fermion and boson Hamiltonians can be cast into this diagonal form. An oracle call only entails computing the diagonal matrix element given the initial state $|\mathcal{F}\rangle$. Clearly, applying sophisticated quantum simulation methods to free theories is overkill. However, we discuss these theories because they are the simplest examples, and because these terms occur in interacting theories where nontrivial methods are necessary.

The enumerator oracle for a diagonal Hamiltonian simply copies any input Fock state to the output register:
\begin{equation}
    O_F:|\mathcal{F}\rangle|0\rangle\rightarrow|\mathcal{F}\rangle|\mathcal{F}\rangle.
\end{equation}
Thus it can be implemented by a single layer of CNOTs, one to copy the state of each qubit (in the computational basis).

In light-front quantization, the Hamiltonian for a free boson of mass $m_B$ in 1+1D is 
\be
    H = m_B^2\sum_{n=1}^K\frac{1}{n}a_n^\dag a_n,
\ee
where the different values of $n$ are light-front momenta, and $K$ is the total light-front momentum (harmonic resolution; see \cref{lf_application}).
The coefficient function for the Hamiltonian is therefore
\be\label{beta_free_boson}
\beta_{\{(0),(0)\}}(n) = \frac{m_B^2}{n},
\ee 
where `0' denotes boson.
Thus in this case the operations required to compute the matrix element are just those to compute a reciprocal.

Recall that our interactions as in \eqref{interaction_eq} (in \cref{interaction_def}) are specified as a pair of lists $\{(q_1,...,q_g),(q_{g+1},...,q_f)\}$, where $(q_1,...,q_g)$ are the outgoing particles and $(q_{g+1},...,q_f)$ are the incoming particles: thus in the present example $\{(0),(0)\}$ means one incoming boson and one outgoing boson.
The Hamiltonian can be rewritten in terms of \eqref{beta_free_boson} as
\be
H = \sum_{n=1}^K \beta_{\{(0),(0)\}}(n) a_n^\dag a_n.
\ee
The matrix element oracle for the free boson field Hamiltonian is
\begin{equation}
\begin{split}
&O_H:|\mathcal{F}\rangle|\mathcal{F}\rangle\rightarrow |\mathcal{F}\rangle |\mathcal{F}\rangle\left|\sum_{n=1}^K\beta_{\{(0),(0)\}}(n)w_n\right\rangle,\\
&O_H:|\mathcal{F}\rangle|\mathcal{F}'\rangle\rightarrow |\mathcal{F}\rangle |\mathcal{F}'\rangle|0\rangle,
\end{split}
\end{equation}
where $w_n$ is the occupation of the mode with light-front momentum $n$ in $|\mathcal{F}\rangle$, and $|\mathcal{F}'\rangle\neq|\mathcal{F}\rangle$.

The Hamiltonian for the Dirac field in 1+1D light-front quantization is 
\be H = m_F^2\sum_{n=1}^K \frac{1}{n}(b_n^\dag b_n+d_n^\dag d_n), \ee
where $m_F$ is the fermion/antifermion mass.
The coefficient function for each interaction  is 
\be  \beta_{\{(1),(1)\}}(n)=\beta_{\{(2),(2)\}}(n) = \frac{m_F^2}{n},  \ee
where `1'  denotes fermion and `2' denotes antifermion. Rewriting the Hamiltonian in terms of these gives
\be H= \sum_{n=1}^K 
\left(\beta_{\{(1),(1)\}}(n)b_n^\dag b_n+\beta_{\{(2),(2)\}}(n)d_n^\dag d_n\right). \ee 
The matrix element oracles for the two interactions in the Dirac field Hamiltonian are thus identical to the matrix element oracle for the free boson field, replacing the coefficient functions and occupation numbers with those corresponding to fermions and antifermions for the first and second interactions, respectively.

In equal-time quantization, the free Hamiltonian in second-quantized form looks similar to that in light-front quantization.
The only difference is that the sum runs over positive and negative momenta, and the coefficient function is given by
\begin{equation}
\label{final_free_coeff}
    \beta_{\{(i),(i)\}}(n)= \frac{1}{\sqrt{m^2+n^2}}=\frac{1}{\omega_n},
\end{equation}
where $i=0$, $1$, or $2$ for the boson, fermion, or antifermion interactions, and $m$ is the mass of the particle (scaled by the box size $L$).
Thus in this case the operations required to compute the matrix element are a square, a sum, a square-root, and a reciprocal.

\subsection{$\lambda \phi^4$ theory}

In light-front quantization, the $\lambda\phi^4$ theory in 1+1D has the Hamiltonian~\cite{PhysRevD.36.1141}
\begin{equation} 
\label{phi4ham}
    H=H_0+H_I,
\end{equation}
where
\begin{equation}
    H_0=\sum_{n}\frac{1}{n}a_n^\dag a_n\left(m^2+\frac{\lambda}{4\pi}\frac{1}{2}\sum_k\frac{1}{k}\right)
\end{equation}
and
\begin{equation}
\label{phi4_interacting_part}
\begin{split}
    H_I=&\frac{1}{4}\frac{\lambda}{4\pi}\sum_{klmn}\frac{a_k^\dag a_l^\dag a_m a_n}{\sqrt{klmn}}\delta_{m+n,k+l}\\
    +&\frac{1}{6}\frac{\lambda}{4\pi}\sum_{klmn}\left(\frac{a_k^\dag a_la_ma_n+a_k^\dag a_l^\dag a_m^\dag a_n}{\sqrt{klmn}}\right)\delta_{k,m+n+l},
\end{split}
\end{equation}
where $\lambda$ is the coupling constant, and $H_0,H_I$ are the free and interacting parts of the Hamiltonian, respectively.
The sums are over light-front momenta in the range $[1,K]$. We can treat the free part of the Hamiltonian by the methods of~\cref{freeoracle}. In this section we focus on the interacting part of the Hamiltonian, given in \eqref{phi4_interacting_part}.

$H_I$ is composed of three interactions, corresponding to the ladder operator monomials $a_k^\dagger a_la_ma_n$, $a_k^\dagger a_l^\dagger a_ma_n$, and $a_l^\dagger a_m^\dagger a_n^\dagger a_k$ (summed over the momenta).
In our interaction notation as in \cref{interaction_def}, these are written $\{(0),(0,0,0)\}$, $\{(0,0),(0,0)\}$, and $\{(0,0,0),(0)\}$, respectively.
Reading off from \eqref{phi4_interacting_part}, the coefficient functions are given by
\begin{equation}
\label{phi4_lf_coeff_1}
    \beta_{\{(0,0),(0,0)\}}(k,l,m,n)=\frac{\lambda}{16\pi\sqrt{klmn}},
\end{equation}

\begin{equation}
\label{phi4_lf_coeff_2}
    \beta_{\{(0),(0,0,0)\}}(k,l,m,n)=\frac{\lambda}{24\pi\sqrt{klmn}},
\end{equation}
and
\begin{equation}
\label{phi4_lf_coeff_3}
    \beta_{\{(0,0,0),(0)\}}(k,l,m,n)=\frac{\lambda}{24\pi\sqrt{klmn}}.
\end{equation}
Note that the delta functions that enforce momentum conservation are not included in the coefficient functions, because momentum conservation is enforced at an earlier step in the algorithm than computation of matrix elements.
These are the entirety of the inputs needed to specify our oracle implementations.
Rewriting the Hamiltonian in terms of the coefficient functions gives
\begin{equation}
\begin{alignedat}{8}
    H_I = &\sum_{klmn}  \beta_{\{(0,0),(0,0)\}}(k,l,m,n) &&a_k^\dag a_l^\dag a_m a_n \\
    +&\sum_{klmn} \beta_{\{(0,0,0),(0)\}}(k,l,m,n) &&a_k^\dag a_la_ma_n \\
    +&\sum_{klmn}\beta_{\{(0),(0,0,0)\}}(k,l,m,n)&&a_k^\dag a_l^\dag a_m^\dag a_n,
\end{alignedat}
\end{equation}
where the sums run over momentum-conserving combinations of $k,l,m,n\in\{1,2,...,K\}$ for total light-front momentum $K$.

In equal-time quantization, the interacting part of the $\lambda \phi^4$ Hamiltonian in 1+1D is
\begin{multline}
\label{phi4_ET}
H_I=\frac{\lambda}{4!}\sum_{k,l,p,f}\frac{1}{\sqrt{16\omega_p\omega_l\omega_k\omega_f}}[\\
a_pa_ka_la_f\delta_{-f-l,k+p}+a_p^\dag a_k^\dag a_l^\dag a_f^\dag \delta_{l+f,-k-p}\\
+4a_f^\dag a_pa_ka_l\delta_{f,k+l+p}+6a_ka_l\delta_{f,p}\delta_{k,-l}\\
+6a_l^\dag a_f^\dag a_pa_k\delta_{l+f,k+p}\\
+6a_k^\dag a_l^\dag\delta_{f,p}\delta_{k,-l}+4a_k^\dag a_l^\dag a_f^\dag a_p\delta_{l+f+k,p}].
\end{multline}
From this, we can read off the coefficient functions:
\begin{equation}
\label{phi4_et_coeff}
\begin{alignedat}{8}
&\beta_{\{(0,0,0,0),()\}}(k,l,f,p)= \frac{1}{\sqrt{16\omega_p\omega_l\omega_k\omega_f}},\\
&\beta_{\{(),(0,0,0,0)\}}(k,l,f,p) = \frac{1}{\sqrt{16\omega_p\omega_l\omega_k\omega_f}},\\
&\beta_{\{(0,0,0),(0)\}}(k,l,f,p)= \frac{4}{\sqrt{16\omega_p\omega_l\omega_k\omega_f}},\\
&\beta_{\{(0,0),()\}}(k,l)=\sum_{f,p}\frac{6}{\sqrt{16\omega_p\omega_l\omega_k\omega_f}}\delta_{f,p},\\
&\beta_{\{(0,0),(0,0)\}}(k,l,f,p)= \frac{6}{\sqrt{16\omega_p\omega_l\omega_k\omega_f}},\\
&\beta_{\{(),(0,0)\}}(k,l)= \sum_{f,p}\frac{6}{\sqrt{16\omega_p\omega_l\omega_k\omega_f}}\delta_{f,p},\\
&\beta_{\{(0),(0,0,0)\}}(k,l,f,p )=\frac{4}{\sqrt{16\omega_p\omega_l\omega_k\omega_f}}.
\end{alignedat}
\end{equation}
Note that the coefficient functions for the interactions having only two external lines ($\beta_{\{(),(0,0)\}}$ and $\beta_{\{(0,0),()\}}$) are only functions of the momenta of those lines ($k$ and $l$). 
For the sake of brevity, we will omit rewriting the Hamiltonian explicitly in terms of the coefficient functions going forward, as we hope this correspondence has been made clear from the examples above.
As before, the coefficient functions together with their associated interactions are the inputs required to define the Hamiltonian oracles.

The log-local gate counts for an explicit implementation of the oracles for $\lambda\phi^4$ in both light-front and equal-time quantization are given in \cref{fig:phi4_gate_count}.
The gate counts are given as a function of $K$.
The equal-time counts are obtained by choosing single-particle momentum cutoffs $[-\Lambda,\Lambda]$ for $\Lambda=\lceil K/2\rceil-1$, so that the total number of lattice points in the equal-time simulation would be
\begin{equation}
\label{lambda_in_terms_of_K}
    2\Lambda+1=2\lceil K/2\rceil-1=
    \begin{cases}
        K-1\quad\text{for even $K$,}\\
        K\quad\text{for odd $K$,}
    \end{cases}
\end{equation}
for a fair comparison to $K$ lattice points in the light-front simulation.
For the equal-time counts, we also have to impose a cutoff on the number $I$ of distinct occupied modes (which is \emph{a priori} arbitrary), which we choose to be $K$, again to provide a conservative comparison to the light-front simulation (in which the number of distinct occupied modes is in fact smaller still at $O(\sqrt{K})$~\cite{kreshchuk20a}).

The number of log-local operations to implement the enumerator oracle is given by \eqref{cost_scaling}: $O\left(I^{f-g}+\Lambda_\text{max}^{dg}\right)$.
Since the matrix-element oracle requires only logarithmically-many additional gates, implementing it does not change the asymptotic scaling, as noted above.
Since $I=K$ for equal-time, this number of log-local operations becomes $O(K^4)$ due to the interactions that have four incoming particles ($f-g=4$) in \eqref{phi4_ET}.
For light-front, $I=\sqrt{K}$ so the term $I^{f-g}$ no longer dominates, but the second term gives $O(K^3)$ due to the interactions that have three outgoing particles ($g=3$) in \eqref{phi4_interacting_part} (since $\Lambda_\text{max}=K$ in this case).

These costs are indeed what we see in the log-log plot \cref{fig:phi4_gate_count}, which shows the exact log-local gate counts for light-front and equal-time.
To illustrate the asymptotic behaviors, \cref{fig:phi4_gate_count} also plots $K^3$ and $K^4/5$.
Since we expect the light-front cost to be $O(K^3)$ and the equal-time cost to be $O(K^4)$, as discussed above, the slopes of the data should approach those of the plotted lines on the log-log plot, which is what we see.

The extra operations required to compute the matrix elements for light-front quantization are those required to evaluate \eqref{phi4_lf_coeff_1}, \eqref{phi4_lf_coeff_2}, and \eqref{phi4_lf_coeff_3}, namely products, square-roots, and reciprocals.
The extra operations required to compute the matrix elements for light-front quantization are those required to evaluate \eqref{phi4_et_coeff} for the frequencies $\omega_n$ defined as in \eqref{final_free_coeff}.
Hence they require squares, sums, products, square-roots, and reciprocals.

\begin{figure}
\centering
\includegraphics[width=\columnwidth]{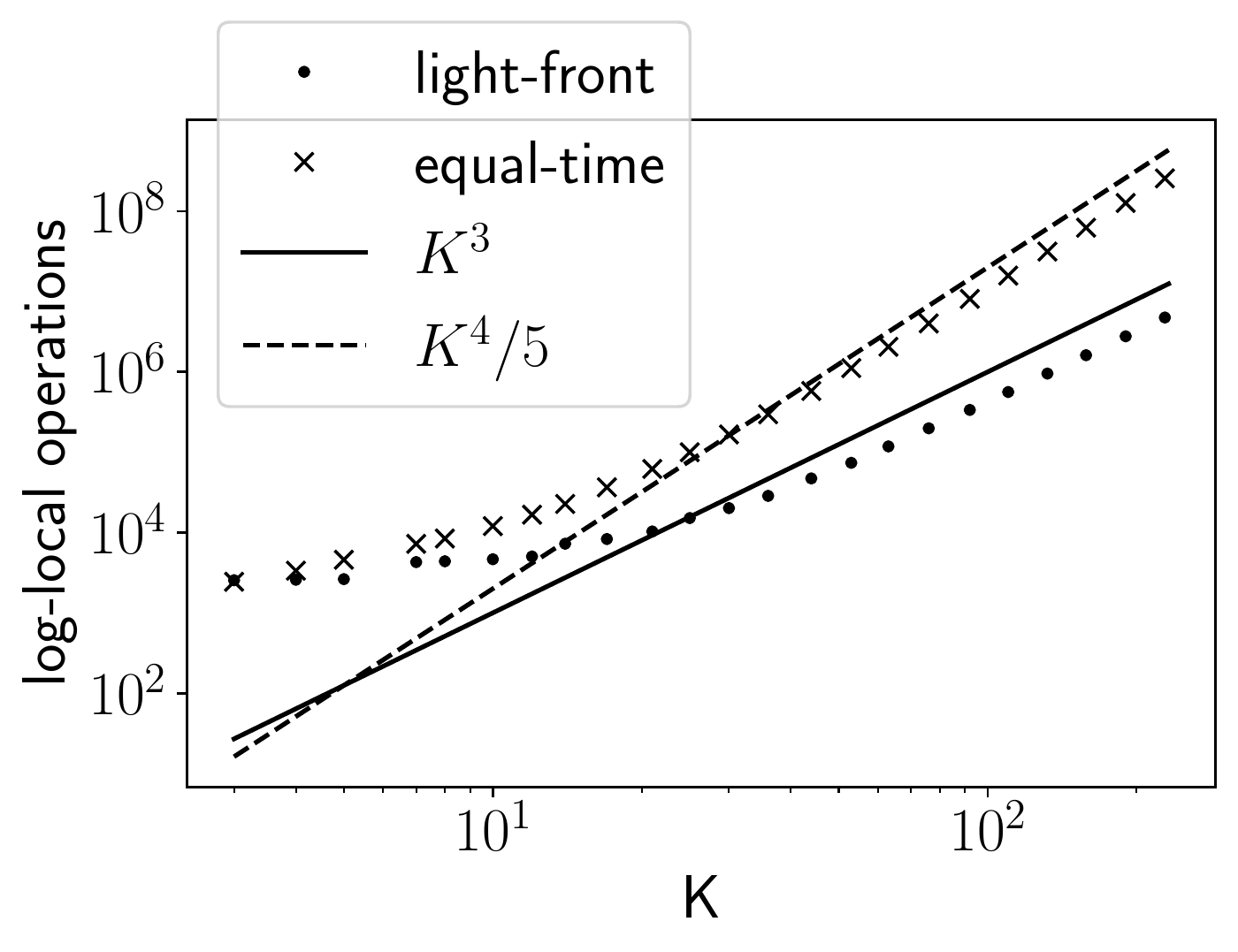}
\caption{Gate counts (in log-local operations) to implement oracles for $\lambda\phi^4$ theory in 1+1D. The equal-time cutoffs are $[-\Lambda,\Lambda]$, for $\Lambda$ defined in terms of $K$ by \eqref{lambda_in_terms_of_K}. The exact gate counts for light-front and equal-time quantization are given by the points and crosses, respectively. The solid line $K^3$ and dashed line $K^4/5$ are included to illustrate that the datapoints are indeed converging to their expected asymptotic scalings of $O(K^3)$ for light-front and $O(K^4)$ for equal-time.}
\label{fig:phi4_gate_count}
\end{figure}

\subsection{Massive Yukawa model}

\begin{figure}
\centering
\includegraphics[width=\columnwidth]{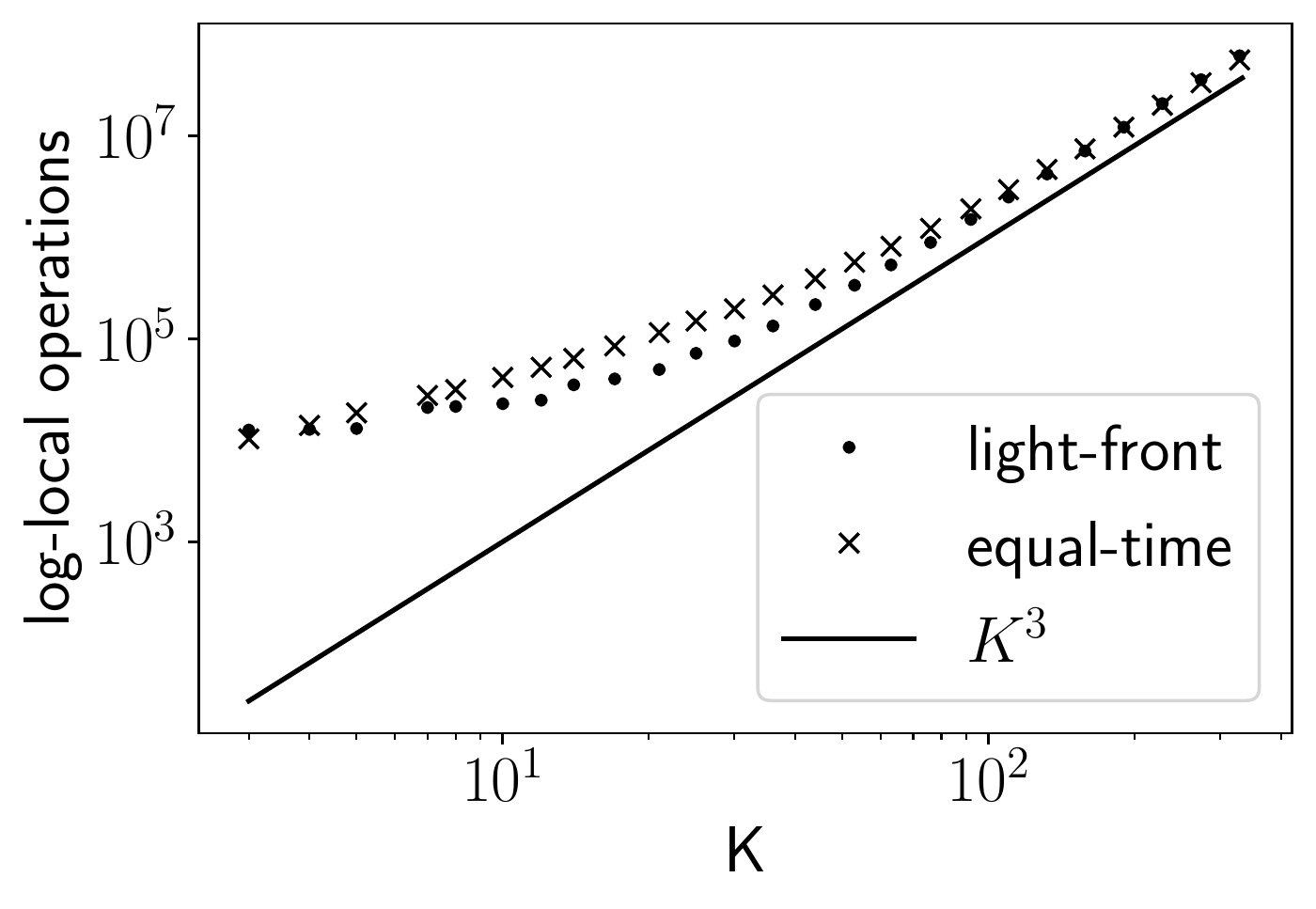}
\caption{Gate counts (in log-local operations) to implement oracles for the massive Yukawa model. The equal-time cutoffs are $[-\Lambda,\Lambda]$, for $\Lambda$ defined in terms of $K$ by \eqref{lambda_in_terms_of_K}. The exact gate counts for light-front and equal-time quantization are given by the points and crosses, respectively. The line $K^3$ is included to illustrate that the datapoints are indeed converging to their expected asymptotic scaling of $O(K^3)$.} 
\label{fig:yukawa_gate_count}
\end{figure}

For the massive Yukawa model in 1+1D light-front quantization, we only write the interaction Hamiltonian~\cite{PhysRevD.32.1993}.
This is composed of the so called \emph{vertex}, \emph{seagull}, and \emph{fork} terms:
\be
\label{lf_yukawa_ham}
H_I =H_V+H_S+H_F.
\ee
The first term, $H_V$, is
\begin{equation} 
\begin{split}
    H_V=g&m_F\sum_{k,l,m}\\
    \bigg[&\frac{2}{(k+l)\sqrt{l}}(b_k^\dag b_m a_l^\dag+b_m^\dag b_ka_l)\\
    &+\frac{2}{(k+l)\sqrt{l}}(d_k^\dag d_ma_l^\dag+d_m^\dag d_ka_l)\\
    &+\frac{2}{(k-m)\sqrt{m}}(b_kd_la_m^\dag+d_l^\dag b_k^\dag a_m)\bigg]\delta_{k+l,m},
\end{split}
\end{equation}
where the $a^{(\dagger)}$ are boson ladder operators, $b^{(\dag)}$ are fermion ladder operators, and $d^{(\dag)}$ are antifermion ladder operators. 
The resulting coefficient functions for $H_V$ are
\begin{equation}
\begin{split}
    & \beta_{\{(1),(1,0)\}}(k,l,m) = \beta_{\{(1,0),(1)\}}(k,l,m) \\
    & = \beta_{\{(2,2),(0)\}}(k,l,m) \\
    & = \beta_{\{(2,0),(2)\}}(k,l,m) = \frac{2gm_F}{(k+l)\sqrt{l}},
\end{split}
\end{equation}
\begin{equation}
\begin{split}
    & \beta_{\{(1,2),(0)\}}(k,l,m) = \beta_{\{(0),(1,2)\}}(k,l,m) \\
    & = \frac{2gm_F}{(k-m)\sqrt{m}},
\end{split}
\end{equation}
where in the subscripts, `0' denotes boson, `1' denotes fermion, and `2' denotes antifermion.

The second term in \eqref{lf_yukawa_ham}, $H_S$, is
\begin{equation}
\begin{split}
    H_S = g^2\sum_{k,l,m,n}\bigg[&\frac{1}{m-k}(d_kb_la_m^\dag a_n^\dag +b_l^\dag d_k^\dag a_n a_m)\\
    &+\frac{2}{k-n}b_k^\dag b_l a_m^\dag a_n\\
    &+\frac{2}{k-n}d_k^\dag d_la_m^\dag a_n\bigg]\frac{\delta_{k+l,m+n}}{\sqrt{mn}},
\end{split}
\end{equation}
resulting in the following coefficient functions:
\begin{equation}
\begin{split}
    & \beta_{\{(1,0),(1,0)\}}(k,l,m,n) \\
    & = \beta_{\{(2,1),(2,1)\}}(k,l,m,n)=\frac{2g^2}{(k-n)\sqrt{mn}},
\end{split}
\end{equation}
\begin{equation}
\begin{split}
    & \beta_{\{(2,1),(0,0)\}}(k,l,m,n) \\
    & = \beta_{\{(0,0),(2,1)\}}(k,l,m,n)=\frac{g^2}{(m-k)\sqrt{mn}}.
\end{split}
\end{equation}

The third and final term in \eqref{lf_yukawa_ham}, $H_F$, is
\begin{equation} 
\begin{split}
    H_F= g^2&\sum_{k,l,m,n}\\
    \bigg[&\frac{1}{(k+l)\sqrt{lm}}(b_k^\dag b_na_l^\dag a_m^\dag+b_n^\dag b_ka_ma_l)\\
    &+\frac{1}{(k+l)\sqrt{lm}}(d_k^\dag d_n a_l^\dag a_m^\dag+d_n^\dag d_k a_m a_l)\\
    &+\frac{2}{(k-n)\sqrt{ln}}b_k^\dag d_m^{\dag}a_l^\dag a_n\\
    &+\frac{2}{(k-n)\sqrt{ln}}d_mb_ka_n^\dag a_l\bigg]\delta_{k+l+m,n},
\end{split}
\end{equation}
resulting in the following coefficient functions:
\begin{equation}
\begin{split}
    & \beta_{\{(1)(0,0,1)\}}(k,l,m,n) \\
    & = \beta_{\{(0,0,1),(1)\}}(k,l,m,n) \\
    & = \beta_{\{(2)(0,0,2)\}}(k,l,m,n) \\
    & = \beta_{\{(0,0,2),(2)\}}(k,l,m,n) = \frac{g^2}{(k+l)\sqrt{lm}},
\end{split}
\end{equation}
\begin{equation}
\begin{split}
    &\beta_{\{(0),(0,1,2)\}}(k,l,m,n) \\
    & = \beta_{\{(0,1,2),(0)\}}(k,l,m,n) = \frac{2g^2}{(k-n)\sqrt{ln}}.
\end{split}
\end{equation}
 
In 1+1D equal-time quantization, the interacting part of the Yukawa Hamiltonian with free field expansion is:
\begin{multline}
    H_{I}=\sum_{l,k,p}\frac{1}{\sqrt{2\omega_k}} \frac{1}{\sqrt{2\omega_p}} \frac{1}{\sqrt{2\omega_l}}\sum_{\gamma,s}\bigg[\\
    c_l^{s\dag}c_k^\gamma a_p\bar{\mu}^s(l)\mu^\gamma(k)\delta_{l,k+p}+c_l^{s\dag}d_k^{\gamma\dag}a_p\bar{\mu}^s(l)\nu^\gamma(k)\delta_{l+k,p}\\
    +d_l^{s}c_k^\gamma a_p\bar{\nu}^s(l)\mu^\gamma(k)\delta_{l+k+p,0}-d_k^{\gamma\dag}d_l^{s}a_p\bar{\nu}^s(l)\nu^\gamma(k)\delta_{k,l+p}\\
    +c_l^{s\dag}c_k^\gamma a_p^\dag\bar{\mu}^s(l)\mu^\gamma(k)\delta_{k,l+p}+c_l^{s\dag}d_k^{\gamma\dag}a_p^\dag\bar{\mu}^s(l)\nu^\gamma(k)\delta_{l+k+p,0}\\
    +d_l^{s}c_k^\gamma a_p^\dag\bar{\nu}^s(l)\mu^\gamma(k)\delta_{l+k,p}-d_k^{\gamma\dag}d_l^{s}a_p^\dag \bar{\nu}^s(l)\nu^\gamma(k)\delta_{l,k+p}\bigg],
\end{multline}
where $\mu$ is the fermion spinor, $\nu$ is the antifermion spinor, and $s$ and $\gamma$ are the spin indices.
This leads to the following coefficient functions:
\begin{equation}
\begin{split}
    &\beta_{\{(0,1),(1\}}(k,l,p,s,\gamma) \\
    & = \beta_{\{(1),(0,1)\}}(k,l,p,s,\gamma) = \frac{1}{\sqrt{8\omega_k\omega_p\omega_l}}\bar{\mu}^s(l)\mu^\gamma(k),
\end{split}
\end{equation}
\begin{equation}
\begin{split}
    & \beta_{\{(0),(1,2)\}}(k,l,p,s,\gamma) \\
    & = \beta_{\{(),(0,1,2)\}}(k,l,p,s,\gamma) = \frac{1}{\sqrt{8\omega_k\omega_p\omega_l}}\bar{\mu}^s(l)\nu^\gamma(k),
\end{split}
\end{equation}
\begin{equation}
\begin{split}
    & \beta_{\{(1,2),(0)\}}(k,l,p,s,\gamma) \\
    & = \beta_{\{(0,1,2),()\}}(k,l,p,s,\gamma) = \frac{1}{\sqrt{8\omega_k\omega_p\omega_l}}\bar{\nu}^s(l)\mu^\gamma(k),
\end{split}
\end{equation}
\begin{equation}
\begin{split}
    & \beta_{\{(0,2),(2)\}}(k,l,p,s,\gamma) \\
    & = \beta_{\{(2),(0,2)\}}(k,l,p,s,\gamma) = \frac{1}{\sqrt{8\omega_k\omega_p\omega_l}}\bar{\nu}^s(l)\nu^\gamma(k).
\end{split}
\end{equation}
Note that we have expanded our set of arguments of the coefficient functions to include the spin indices $s,\gamma$.
If we instead wished to obtain interactions exactly as defined in \cref{interaction_def}, we could let each of the above coefficient functions split into four functions, one for each of the pairs of values for the spin indices, but this would just become unwieldy, and there is no harm in including the spin indices as arguments.

The log-local gate counts for implementing the oracles for the Yukawa interaction in both light-front and equal-time quantization are given in \cref{fig:yukawa_gate_count}, with the equal-time cutoff $\Lambda=\lceil K/2\rceil-1$ as for the $\phi^4$ theory, above.
Recall that \eqref{cost_scaling} gives the scaling of the number of log-local operations required to implement the oracles.
Since $\Lambda=O(K)$ as we just discussed, and $I<O(K)$ in both light-front and equal-time, from \eqref{cost_scaling} we see that in both light-front and equal-time the most costly interactions to simulate are those with three outgoing particles.
This gives a cost in log-local operations of $O(K^3)$, which is indeed what we see in \cref{fig:yukawa_gate_count}.

\section{Beyond the plane wave momentum basis}
\label{beyond_momentum}

We have demonstrated how, given a second-quantized Hamiltonian in the plane wave momentum representation of a field theory, we can implement the oracle unitaries necessary to apply sparsity-based simulation methods.
Our methods extend to any second-quantized Hamiltonian containing a fixed number of interactions, even if it is not expressed in the plane wave momentum basis.
All that is required is that for each interaction, it is possible to efficiently enumerate all possible sets of outgoing particles given a particular set of incoming particles, and to efficiently compute the matrix element given the incoming and outgoing particles.
In the plane wave momentum representation, we used momentum conservation for the former task: given a set of incoming particles, we can add up their momenta to obtain the total transferred momentum, and then enumerate all possible allocations of this momentum amongst the outgoing particles.

However, more generally the sets of outgoing particles can always be enumerated in polynomial time as long as the number of outgoing particles is fixed and there are only polynomially-many possible states for each particle.
Here polynomial means polynomial in whatever problem parameter governs the asymptotic scaling.
If $g$ is the number of outgoing particles and $P$ is an upper bound on the number of states that each outgoing particle may take, then $P^g$ is an upper bound on the number of distinct sets of outgoing particles from a particular set of incoming particles.
We could apply this argument to the plane wave momentum basis case, and it would lead to an efficient algorithm, but with worse scaling than the one we presented above, since it would overcount the possible outgoing states.
This illustrates that using momentum conservation at the level of enumerating outgoing states was really an additional constraint that we imposed in order to save resources, rather than an intrinsically necessary part of the algorithm.

Hence there is no problem with extending our algorithm to a non-momentum basis as long as it is possible to efficiently calculate the matrix element between two Fock states.
If there is no conserved quantity that constrains the outgoing particles, then we can enumerate all of the possible sets of outgoing particles as described above.
If there is a conserved quantity (or more than one), then just as for momentum conservation we can compute its value for the incoming particles and then only enumerate sets of outgoing particles that conserve it.
But to reiterate, either of these approaches is efficient; choosing whether or not to exploit a conserved quantity simply changes the details of the scaling.
Hence, although our main presentation focused on the plane wave momentum basis, we can apply our methods to a wide variety of theories expressed in other bases, in quantum chemistry, condensed matter physics, and quantum field theory, including basis light-front quantization~\cite{varybasis,kreshchuk2020light,kreshchuk2020blfq}.

\section{Conclusion}
\label{conclusion}

In this paper, we presented implementations of the Hamiltonian oracles for second-quantized Hamiltonians of theories including bosons and fermions.
We focused on the plane wave momentum basis, but the methods we described generalize to any second-quantized Hamiltonian as long as it only contains polynomially-many terms (monomials in the creation and annihilation operators), and as long as the coefficients of the terms can be computed efficiently.
These oracle implementations are the necessary inputs to any of the large collection of simulation techniques for sparse Hamiltonians~\cite{aharonov03a,childs03a,berry07a,childs10a,berry12a,berry14a,berry15a,berry15b,low17a,low19a,berry20a}.
The generality of our algorithms means that for some specific field theories, it is likely possible to develop algorithms that are tailored to the structure of the theories and outperform our methods (see \cref{prior_work}, for example).
However, our goal was to provide a general-purpose tool, and thus to establish that for second-quantized Hamiltonians satisfying only the modest constraints stated above, efficient quantum simulation by optimal sparsity-based methods is possible.

~
\begin{acknowledgements}
W.~M.~K. acknowledges support from the National Science Foundation, Grant No. DGE-1842474.
M.~K. acknowledges support from DOE HEP Grant No. DE-SC0019452.
This work was supported by the NSF STAQ project (PHY-1818914), and by the U.S. Department of Energy, Office of Science, National Quantum Information Science Research Centers, Quantum Systems Accelerator (QSA).
\end{acknowledgements}

\bibliography{references}

\appendix

\end{document}